\documentclass[aps, prl, floatfix,
  reprint, 
  superscriptaddress, 
  longbibliography
]{revtex4-2} 
\usepackage{graphicx, xcolor, 
adjustbox
} 
\graphicspath{ {./Figure/} }
\usepackage{physics, amsmath, amssymb, amsthm, bm} 
\usepackage{mathtools, bm} 
\usepackage{booktabs}
\usepackage{hyperref}
\usepackage{url}
\usepackage{tikz}
\hypersetup{
	colorlinks=true,
    citecolor=blue,
    linkcolor=blue,
    urlcolor=blue,
	bookmarksnumbered=true,
	pdfborder={0 0 0},
	bookmarkstype=toc,
	pdfauthor={Shozo Yamada, Akihiro Hokkyo, and Masahito Ueda},
	pdftitle={Burst
    }
}
\usepackage[whole]{bxcjkjatype} 
\usepackage{enumitem}

\theoremstyle{Theorem}
\newtheorem{dfn}{Definition}
\newtheorem{prop}[dfn]{Proposition}
\newtheorem{lem}[dfn]{Lemma}

\newtheorem{rem}[dfn]{Remark}

\renewcommand\paragraph[1]{%
\par\emph{#1---}\kern2pt\relax\ignorespaces}

\begin{document}
\title{
  Tunable many-body burst in isolated quantum systems
}
\author{Shozo Yamada}
  \email{yamada@cat.phys.s.u-tokyo.ac.jp}
  \affiliation{Department of Physics, The University of Tokyo, 7-3-1 Hongo, Bunkyo-ku, Tokyo 113-0033, Japan}
\author{Akihiro Hokkyo}
  \email{hokkyo@cat.phys.s.u-tokyo.ac.jp}
  \affiliation{Department of Physics, The University of Tokyo, 7-3-1 Hongo, Bunkyo-ku, Tokyo 113-0033, Japan}
\author{Masahito Ueda}
  \email{ueda@phys.s.u-tokyo.ac.jp}
  \affiliation{Department of Physics, The University of Tokyo, 7-3-1 Hongo, Bunkyo-ku, Tokyo 113-0033, Japan}
  \affiliation{Fundamental Quantum Science Program (FQSP), TRIP Headquarters, RIKEN, Wako, Saitama 351-0198, Japan}
  \affiliation{Institute for Physics of Intelligence, The University of Tokyo, 7-3-1 Hongo, Bunkyo-ku, Tokyo 113-0033, Japan}
\begin{abstract}

Thermalization in isolated quantum many-body systems can be nonmonotonic depending on an initial state.
We demonstrate that a burst---a transient deviation from thermal equilibrium---can be created at a designated time from a low-entangled initial state for a nonintegrable mixed-field Ising chain, by developing a systematic numerical method to construct such a low-entangled state.
Contrary to the typical spreading of information, the created burst is accompanied by a slow or even negative entanglement growth.
We analytically show that a burst becomes probabilistically rare after a long time.
Our results suggest that a nonequilibrium state is maintained for an appropriately chosen initial state until quantum scrambling becomes dominant.
Our predictions can be tested by quench experiments with programmable quantum simulators.

\end{abstract}
\maketitle

\paragraph{Introduction}
Ever since von Neumann's seminal work~\cite{Neumann1929_beweis}, how irreversible macroscopic thermalization emerges from reversible microscopic laws of quantum mechanics has been actively debated. 
This fundamental problem of quantum thermalization has attracted wide interest~\cite{Polkovnikov2011_colloquium, Eisert2015_quantum, Gogolin2016_equilibration, DAlessio2016_quantum, Mori2018_thermalization} due in large part to experimental advances in ultracold atoms~\cite{Kinoshita2006_quantum, Trotzky2012_probing, Gring2012_relaxation, Langen2013_local, Schreiber2015_observation, Kaufman2016_quantum, Bernien2017_probing,Scherg2021_observing}, trapped ions~\cite{Richerme2014_non-local, Smith2016_many-body,Clos2016_time-resolved,Neyenhuis2017_observation}, and superconducting qubits~\cite{Neill2016_ergodic,Blok2021_quantum,Chen2021_observation, Zhu2022_observation,Zhang2023_many-body, Wang2025_exploring}.
In these experiments, nonequilibrium dynamics is usually explored via quantum quench: the system is prepared in a low-entangled pure state, such as a ground state~\cite{Trotzky2012_probing,Schreiber2015_observation,Kaufman2016_quantum,Bernien2017_probing,Scherg2021_observing} or a state generated by a quantum circuit~\cite{Neill2016_ergodic,Chen2021_observation, Zhu2022_observation,Zhang2023_many-body, Wang2025_exploring}, and subsequently evolves in time according to an interacting Hamiltonian.
A crucial question here is how thermalization is attained from such low-entangled initial states~\cite{Wilming2019_entanglement-ergodic,Bertoni2025_typical,Pilatowsky-Cameo2026_quantum}.

As a fundamental premise of quantum thermalization, the eigenstate thermalization hypothesis (ETH) was proposed~\cite{Deutsch1991_quantum, Srednicki1994_chaos,Rigol2008_thermalization}.
While exceptions such as many-body localization~\cite{Basko2006_metal-insulator, Schreiber2015_observation, Imbrie2016_many-body, Smith2016_many-body}, quantum many-body scars~\cite{Bernien2017_probing, Shiraishi2017_systematic,Turner2018_weak,Zhang2023_many-body}, and Hilbert space fragmentation~\cite{Sala2020_ergodicity,Khemani2020_localization,Scherg2021_observing,Wang2025_exploring} have been discovered, the validity of the ETH in nonintegrable systems has widely been confirmed numerically~\cite{Steinigeweg2013_eigenstate,Kim2014_testing,Beugeling2014_finite-size,Beugeling2015_off-diagonal, mondaini2017_eigenstate,Garrison2018_single,Sugimoto2021_test}.
The ETH guarantees thermalization in the infinite-time limit from any initial state with subextensive energy fluctuations~\cite{DAlessio2016_quantum}, including the aforementioned low-entangled states. However, it leaves the finite-time dynamics~\cite{Berges2004_prethermalization,Bartsch2009_dynamical, Banuls2011_strong,Short2012_quantum,Goldstein2013_time,Dymarsky2019_mechanism} largely unconstrained. 
In fact, the ETH does not rule out a ``burst''---a transient deviation of an expectation value of an observable from its thermal equilibrium value---which is in sharp contrast to usual monotonic thermalization observed in quench experiments~\cite{Trotzky2012_probing,Kaufman2016_quantum,Chen2021_observation,Zhu2022_observation,Clos2016_time-resolved}.

Indeed, the authors of Ref.~\cite{Knipschild2020_modern} showed that anomalous relaxation dynamics including a burst is consistent with the ETH by constructing specially designed observables, which are generally nonlocal.
Even for local Hamiltonians and observables, a burst can occur via quantum recurrence~\cite{Bocchieri1957_quantum,Percival1961_periodicity}; however, the recurrence time is typically doubly exponential in the system size~\cite{Peres1982_recurrence}.
The experimental realization of the above two possibilities remains a challenge due to nonlocality of observables  or extremely long timescales.

In this Letter, we demonstrate that a burst can be created at a designated time from a low-entangled initial state for a nonintegrable mixed-field Ising chain, using average magnetization as an example (see Fig.~\ref{fig:burst_plot}).
We focus on matrix product states (MPSs)~\cite{Fannes1992_finitely,Perez-Garcia2007_matrix} in one-dimensional quantum systems as operationally accessible candidates for such low-entangled states, and utilize the density matrix renormalization group (DMRG) algorithm~\cite{White1992_density,White1993_density-matrix,Schollwock2005_density-matrix,Schollwock2011_density-matrix} to search for an MPS that creates a burst.
This method enables us to create a tailor-made burst for a given Hamiltonian, observable, burst time, and the bond dimension and energy of an initial state.
We also find a slow or even negative entanglement growth before the burst time.

We numerically show that for a fixed and short burst time, a bond dimension independent of the system size is sufficient to create a large burst even for a large system.
This implies that the burst is experimentally realizable via quantum quenches from ground states~\cite{Verstraete2006_matrix,Perez-Garcia2007_matrix,Hastings2007_area,Sompet2022_realizing} or shallow quantum circuits~\cite{Ran2020_encoding,Shirakawa2024_automatic,Rudolph2024_decomposition,Iaconis2024_quantum}.
Conversely, our analytical argument based on a local random quantum circuit~\cite{Brandao2016_local} shows that a burst becomes exponentially rare in the long-time regime.
Our method allows us to search for the maximal transient deviation from equilibrium accessible via low-entangled initial states.
Such a deviation persists until being overwhelmed by quantum scrambling, even in generic systems that eventually thermalize. This prediction can be tested by quench experiments with programmable quantum simulators.

\begin{figure}
    \centering
    \includegraphics[width=\linewidth]{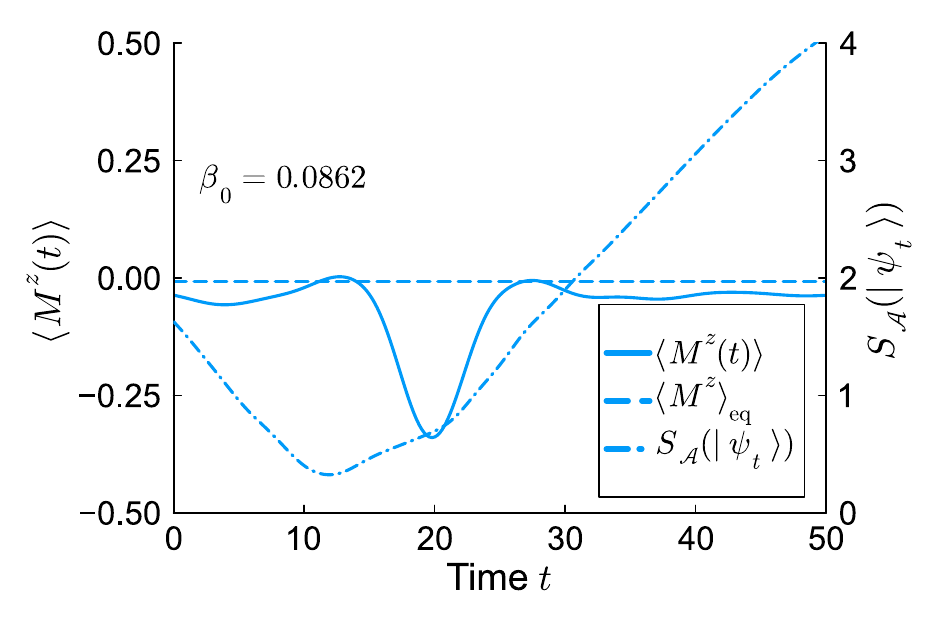}
    \caption{Time evolution of the expectation value $\langle M^z(t)\rangle$ (solid) and that of entanglement entropy $S_\mathcal{A}(\ket{\psi_t})$ (dash-dotted, calculated for the half system $\mathcal{A}=\{1,2,\ldots,\lfloor L/2 \rfloor\}$) with $L=40$ for the mixed-field Ising chain~\eqref{eq:Ising}. A dashed line shows the thermal equilibrium value $\langle M^z\rangle_\mathrm{eq}=-0.00836$. The initial state is obtained with $\tau=20$, $\chi=10$, $\beta=0.1$, and $\lambda_L=72/L^2$. }
    \label{fig:burst_plot}
\end{figure}

\paragraph{Method for creating a burst}
We consider $L$ sites $\{1,2,\ldots,L\}$ on a one-dimensional chain with the open boundary condition.
An MPS can be expressed as $\ket{\psi} = \sum_{\{\sigma_i\}} A_1^{\sigma_1}A_2^{\sigma_2}\cdots A_L^{\sigma_L}\ket{\sigma_1\sigma_2\cdots\sigma_L}$, where $\sigma_i$ denotes a state of site $i$, and $A_i^{\sigma_i}$ is a $D_i\times D_{i+1}$ matrix with $D_1=D_{L+1}=1$.
Entanglement entropy for a subsystem $\mathcal{A}$ is defined as $S_\mathcal{A}(\ket{\psi})\coloneqq-\Tr[\rho_\mathcal{A} \ln\rho_\mathcal{A}]$ with $\rho_\mathcal{A}=\Tr_{\mathcal{A}^c} [\ket{\psi}\bra{\psi}]$, where $\mathcal{A}^c$ denotes the complement of $\mathcal{A}$.
Then, the bond dimension $\chi\coloneqq\max_i D_i$ gives an upper bound on entanglement entropy, where $S_\mathcal{A}(\ket{\psi})\le \ln\chi$ holds for any subsystem $\mathcal{A}=\{1,2,\ldots,i\}$.

We investigate a burst phenomenon starting from an MPS with a bond dimension independent of the system size $L$.
This choice ensures physical relevance: such an MPS is the exact ground state of a local Hamiltonian~\cite{Perez-Garcia2007_matrix,MPS-injectivity}, and can be approximately prepared by a shallow quantum circuit~\cite{Ran2020_encoding,Shirakawa2024_automatic,Rudolph2024_decomposition}.
The bond dimension of the initial state can be tuned to reflect the capability of state preparation in experiments.

Now, let $H$, $O$, and $\tau$ be a Hamiltonian, a local observable, and a burst time, respectively.
We set $\hbar=1$, and denote by $O(t)\coloneqq e^{iHt}Oe^{-iHt}$ a time-evolved observable in the Heisenberg picture.
We use the Trotter decomposition~\cite{Trotter1959_product,Suzuki1976_generalized} to compute the time evolution and the thermal equilibrium value.
To secure computational feasibility, the bond dimension of an MPS or a matrix product operator (MPO)~\cite{Verstraete2004_matrix,Pirvu2010_matrix} is truncated at each Trotter step.
Truncating the MPO of $O(t)$ at each Trotter step is justified for short times due to the locality of $O$~\cite{Alhambra2021_locally}.
We perform numerical calculations using the ITensor library~\cite{Fishman2022_itensor,Fishman2022_codebase}, and confirm that the truncation errors in our setup are indeed small~\cite{MPO-truncation}.

To create a burst of $O$ at $t=\tau$, we consider the following cost function:
\begin{equation}
    H_\text{DMRG} = O(\tau)+\lambda_L(H-\langle H\rangle_\beta)^2.\label{eq:DMRG}
\end{equation}
Here, $\lambda_L~(\ge0)$ is a penalty weight that depends on the system size and $\langle X\rangle_\beta\coloneqq \Tr [X e^{-\beta H}]/\Tr[e^{-\beta H}]$ denotes the thermal equilibrium value of $X$ at the inverse temperature $\beta$.
The cost function $H_\mathrm{DMRG}$ is designed to minimize the expectation value $\langle O(\tau) \rangle$ while suppressing the energy fluctuations around a target energy $\langle H \rangle_\beta$.

We proceed in four steps.
\begin{enumerate}[nolistsep]
    \item Compute the cost function~\eqref{eq:DMRG} as an MPO.\label{m2:heisenberg}
    \item Perform the DMRG optimization on the cost function~\eqref{eq:DMRG} with a bond dimension of at most $\chi$, and obtain an MPS $\ket{\psi_0}$.
    \item Determine the inverse temperature $\beta_0$ from the condition $\langle H \rangle_{\beta_0}=\mel{\psi_0}{H}{\psi_0}$ and obtain the thermal equilibrium value $\langle O \rangle_\mathrm{eq}\coloneqq\langle O \rangle_{\beta_0}$.
    \item Compute the time evolution $\ket{\psi_t}=e^{-iHt}\ket{\psi_0}$ and obtain a burst amplitude  $\langle O \rangle_\mathrm{eq}-\langle O(\tau) \rangle$.
\end{enumerate}
The DMRG algorithm is usually used to find the ground state of a local Hamiltonian; however, we utilize it to find the MPS that minimizes the expectation value of a nonlocal cost function with a restriction on the bond dimension.


\paragraph{Burst phenomenon}
We demonstrate a burst phenomenon for a mixed-field Ising chain:
\begin{equation}
    H = \sum_{i=1}^{L-1} J_z S_i^z S_{i+1}^z + \frac{1}{2}\sum_{i=1}^L(h_x S^x_i + h_z S^z_i)\label{eq:Ising}.
\end{equation}
Here, $S_i^{x,y,z}$ is a spin-1/2 operator at site $i$, and we set the parameters as $(J_z, h_x, h_z) = (1, 0.9045, 0.8090)$.
Time is measured in units of $\hbar/J_z$, and therefore $t$ is dimensionless.
The nonintegrability, or chaotic nature, of this model has been investigated both numerically~\cite{Kim2013_ballistic,Kim2014_testing,Zhou2017_operator} and analytically~\cite{Cao2021_statistical,Chiba2024_proof}.
As an observable $O$, we consider average magnetization $M^z=\sum_{i=1}^L S_i^z/L$.
The case of $M^y=\sum_{i=1}^L S_i^y/L$ is discussed in End Matter.

For the above Hamiltonian and observable, a burst phenomenon with $L=40$ is illustrated in Fig.~\ref{fig:burst_plot}.
The expectation value of $M^z$ shows a pronounced dip at the designated time $\tau=20$.
This sudden burst arises from an MPS of bond dimension $\chi=10$, which is much smaller than $\chi\simeq 2^{L/2}\simeq 10^6$ of typical volume-law states.
This indicates that the initial state $\ket{\psi_0}$ is a simple state with low entanglement, thereby supporting the accessibility of the burst by a quantum quench.
We show in End Matter that the burst can still be observed when $\ket{\psi_0}$ is approximately prepared by a shallow quantum circuit, and that $\ket{\psi_0}$ has a short correlation length, suggesting robustness against local perturbations.

Figure~\ref{fig:burst_plot} also demonstrates that bipartite entanglement entropy shows a slow or even negative increase before the burst time in contrast to a typical linear growth~\cite{Calabrese2005_evolution,Kim2013_ballistic}.
This suggests that local information of the system is temporarily retained in the sense that the reduced density operator of the quantum state remains far from the Gibbs state.

We note that we have successfully tuned not only the burst time and the bond dimension of the initial state, but also its energy.
The spin-down state $\ket{\downarrow}^{\otimes L}$, which is the ground state of $M^z$, corresponds to the inverse temperature $\beta_\mathrm{GS}=-0.2682$.
The initial state $\ket{\psi_0}$ corresponds to $\beta_0=0.0862$, which is closer to the target value $\beta=0.1$ than to $\beta_\mathrm{GS}$.
This result shows that we have successfully tuned the energy of the initial state, while simultaneously creating a large burst.
If we use this result to readjust the target value to $\beta'=0.1143$ via a heuristic ansatz $(\beta'-\beta)/(\beta-\beta_\mathrm{GS})=(\beta-\beta_0)/(\beta_0-\beta_\mathrm{GS})$, we obtain $\beta_0'=0.1001$, which is very close to $\beta=0.1$.


\begin{figure}
    \centering
    \begin{minipage}{\linewidth}
        \centering
        \begin{tikzpicture}
            \node[anchor=north west, inner sep=0] (image) at (0,0) {\includegraphics[width=\linewidth]{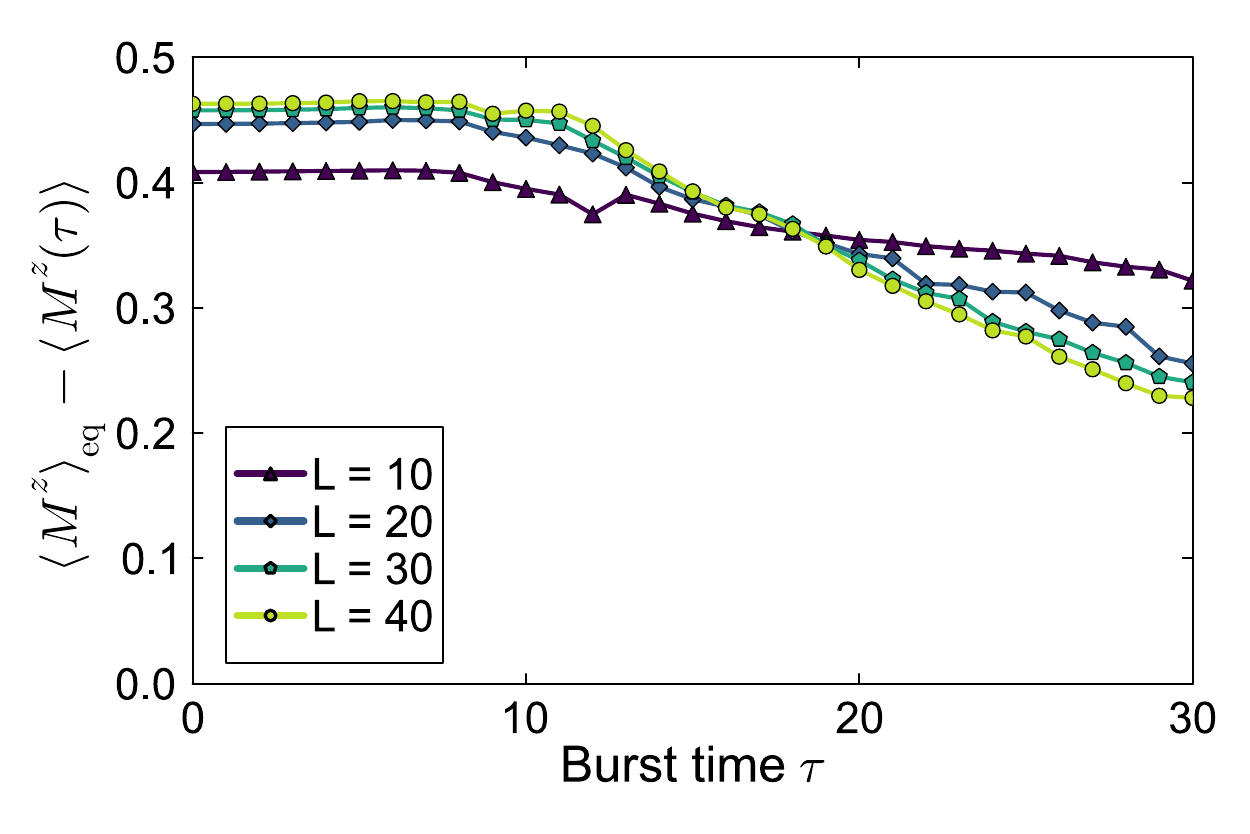}};
            \begin{scope}[x={(image.south east)},y={(image.north west)}]
                \node[anchor=north west] at (0.0, 0.97) {(a)};
            \end{scope}
        \end{tikzpicture}
    \end{minipage}
    \begin{minipage}{\linewidth}
        \centering
        \begin{tikzpicture}
            \node[anchor=north west, inner sep=0] (image) at (0,0) {\includegraphics[width=\linewidth]{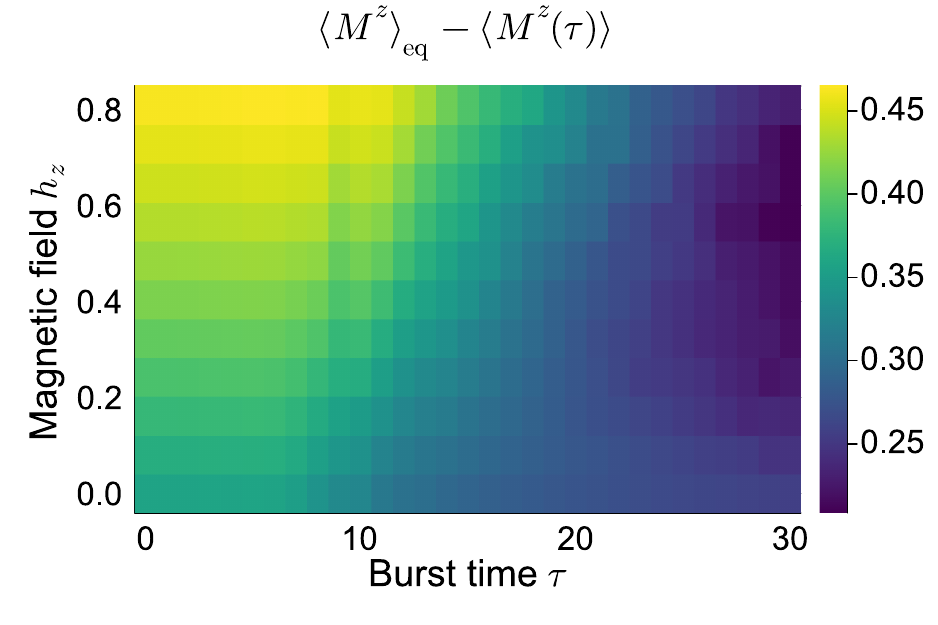}};
            \begin{scope}[x={(image.south east)},y={(image.north west)}]
                \node[anchor=north west] at (0.0, 0.97) {(b)};
            \end{scope}
        \end{tikzpicture}
    \end{minipage}
    \caption{
        Burst amplitude $\langle M^z\rangle_\mathrm{eq}-\langle M^z(\tau) \rangle$ in the mixed-field Ising chain~\eqref{eq:Ising}. Initial states are obtained with $\chi=10$, $\beta=0.1$, and $\lambda_L=72/L^2$. (a) Dependence on the burst time $\tau$ for different system sizes $L$ at a fixed $h_z=0.8090$. (b) Dependence on $\tau$ and the magnetic field $h_z\in [0,0.8090]$ at a fixed $L=40$. }
    \label{fig:Ltau}
\end{figure}

\paragraph{Numerical analysis of the burst amplitude}
We investigate how the burst amplitude $\langle O \rangle_\mathrm{eq}-\langle O(\tau)\rangle$ depends on the burst time $\tau$ with respect to different system sizes $L$ with fixed bond dimension $\chi=10$.
For each $L$ and $\tau$, we follow our method to obtain an initial state $\ket{\psi_0}$ and compute the burst amplitude.
We set $\lambda_L \propto L^{-2}$ since the spectral width $\Delta_O$ does not depend on $L$ for $O=M^z$ while $\Delta_{H^2}$ scales approximately as $L^2$.

The result for the above setting is shown in Fig.~\ref{fig:Ltau}(a).
It shows that the burst amplitude for a fixed $\tau$ remains large with increasing $L$, especially for $\tau\lesssim 18$, where a larger $L$ leads to a larger burst.
This suggests that a large burst can arise on a short-time scale from an MPS with a fixed bond dimension even when the system size is large.
In End Matter, we demonstrate that a burst survives in the thermodynamic limit $L\to \infty$.

The entire time window analyzed here falls within the regime where operator entanglement entropy of the unitary $e^{-iH\tau}$ grows linearly~\cite{Zhou2017_operator}.
Despite this continuous growth of complexity, the burst amplitude remains nearly constant for $\tau\lesssim 11$. 
This suggests that an appropriately chosen initial state can retain local information even amidst ballistic quantum scrambling.
The subsequent gradual decay of the burst amplitude suggests that the complexity of the eigenstates of $O(\tau)$ with low eigenvalues eventually exceeds the representational capacity of an MPS as $\tau$ increases, marking the crossover region beyond which scrambling becomes dominant.

We next investigate how the burst amplitude depends on system parameters.
With $J_z=1$ and $h_x=0.9045$ fixed in the mixed-field Ising Hamiltonian~\eqref{eq:Ising}, we vary the longitudinal magnetic field $h_z$ from the integrable point $h_z=0$ to the nonintegrable regime at $h_z=0.8090$.

Figure~\ref{fig:Ltau}(b) shows that the burst amplitude is initially larger for stronger magnetic fields $h_z$ at a small $\tau$; however, its subsequent decay with $\tau$ is also faster.
The former may be attributed to the energy constraint imposed by the cost function~\eqref{eq:DMRG}; the energy of the spin-down state $\ket{\downarrow}^{\otimes L}$ for the Hamiltonian~\eqref{eq:Ising} is approximately given by $(J_z-h_z)L/4$, which is positive and decreases with increasing $h_z\in[0,0.8090]$.
Therefore, the energy becomes closer to the target value $\langle H\rangle_{\beta=0.1}~(<0)$ for a larger $h_z$, which makes it easier to minimize the two terms of the cost function~\eqref{eq:DMRG} simultaneously and allows for a larger burst amplitude.
The latter suggests that the burst amplitude decays faster as the system becomes more chaotic with $h_z$ moving away from the integrable point $h_z=0$.


\paragraph{Evaluation by a local random circuit}
We consider a more general case with a one-dimensional chain of $L$ sites, each of which has a local dimension $d$.
Here, we consider a random unitary evolution $U$ drawn from a measure $\nu^{*s}_{d,L}$, where in each of the $s$ steps an index $i$ is chosen randomly from $\{1,2,\ldots,L-1\}$, and a unitary $U_{i,i+1}$ drawn from the Haar measure on $\mathbb{U}(d^2)$ acts upon qudits $i$ and $i+1$~\cite{Brandao2016_local}.
We thus approximate the time evolution according to a random time-dependent, 2-local (nearest-neighbor interaction) Hamiltonian~\cite{Poulin2011_quantum}.
Since generic random circuits do not conserve energy, the system thermalizes toward the infinite-temperature state.
Therefore, we quantify the burst amplitude as the deviation from the infinite-temperature average $\langle O\rangle_{\beta=0}=\Tr[O]/d^L$.

First, we consider a burst starting from a fixed initial state $\ket{\psi}$.
In this local random circuit, the probability that the burst is larger than or equal to $\Delta_O a~(a\in (0,1))$ is bounded from above as follows~\cite{Brandao2016_local,Low2009_large,Low2010_pseudo-randomness,Aubrun2026_optimal, comment-zero}:
\begin{align}
    \text{Pr}_{U\sim\nu^{*s}_{d,L}}&\qty[\abs{\mel{\psi}{U^\dag O U}{\psi}-\frac{\Tr[O]}{d^L}}\ge \Delta_O a]\le2\qty(\frac{m}{d^L a^2})^m,\notag\\
    &\mathbb{Z}_{\ge 0}\ni m\le \frac{k}{2},~k=\left\lfloor \qty(\frac{s}{CL^2 d^2\ln d})^{1/11}\right\rfloor,\label{eq:LRProb}
\end{align}
where $C$ is a positive numerical constant. (The proofs for Eqs.~\eqref{eq:LRProb} and \eqref{eq:LRMPS} are given in Supplemental Material.)
This inequality reflects the concentration of measure induced by approximate unitary $k$-designs, which improves as the number of gates $s$ increases~\cite{Brandao2016_local}.

We extend this bound in Eq.~\eqref{eq:LRProb} to the entire manifold $\mathcal{M}_\chi$ of normalized MPSs whose bond dimensions are less than or equal to $\chi$.
We define 
\begin{equation}
    P_{\chi,a}\coloneqq \mathrm{Pr}_{U\sim \nu_{d,L}^{*s}}\qty[\sup_{\ket{\psi}\in\mathcal{M}_\chi}\abs{\mel{\psi}{U^\dagger O U}{\psi}-\frac{\Tr[O]}{d^L}}\ge \Delta_O a]
\end{equation}
as the probability of finding a burst from a state in $\mathcal{M}_\chi$.
For a large system size $L$ and a bond dimension $\chi$ which is independent of $L$, the logarithm of $P_{\chi,a}$ is approximately bounded as
\begin{equation}
    \ln P_{\chi,a}\lesssim -m\ln(md^L)+2(dL\chi^2+m)\ln\frac{dL\chi^2+m}{a}.\label{eq:LRMPS}
\end{equation}

In the local random circuit, we assume that the dimensionless time $\tau$ increases by one unit per $\order{L}$ gates~\cite{Brown2018_second} to imitate the time evolution by a 2-local Hamiltonian.
Therefore, the total number of gates $s$ approximately satisfies $s \propto L\tau$.
Substituting this into Eq.~\eqref{eq:LRProb} gives $k\simeq C'(\tau/L)^{1/11}$ for $\tau\gg L$, where $C'$ is a constant which depends only on $d$.

The asymptotic behavior of the right-hand side of Eq.~\eqref{eq:LRMPS} for fixed $d$, $\chi$, and $a$ is summarized in Table~\ref{tab:LRMPS}.
This quantity does not decay with $L$ for $\tau\ll L(\ln L)^{11}$, which is interpreted as an upper bound on the timescale of a stable burst.
For $L(\ln L)^{11} \ll \tau \ll Ld^{11L}$, choosing $m\simeq (C'/2)(\tau/L)^{1/11}$ shows that the upper bound of $P_{\chi,a}$ decays superexponentially in $L$, with a decay rate that increases with $\tau$.
For $Ld^{11L}\ll \tau$, $m\simeq d^L a^2/e$ gives a tighter bound which decays double-exponentially with $L$.
These results altogether suggest that in the long-time regime, a large burst cannot be created from any MPS with a fixed bond dimension for most time-dependent local Hamiltonians.

We have employed local random circuits to leverage rigorous results of approximate $k$-designs, noting that this setting corresponds to time-dependent Hamiltonians.
However, our burst-creating methods are still applicable even in the time-dependent setting.
Provided that similar rigorous bounds can be established for circuits with symmetries~\cite{Li2024_designs,Hearth2025_unitary,Mitsuhashi2025_unitary} and ultimately extended to energy conservation, the probabilistic evaluation of a burst would follow analogously.
We can also establish a similar bound on entanglement entropy by using the concentration of measure for approximate $k$-designs~\cite{Brandao2016_local,Low2009_large}.

\begin{table}
\centering
    \caption{Choice of $m$ and the leading term (LT) on the right-hand side of Eq.~\eqref{eq:LRMPS} for three regions of the burst time $\tau$. The integer $m$ is chosen to satisfy $m\le k/2$, where $k\simeq C'(\tau/L)^{1/11}$ holds for some constant $C'$ when $\tau\gg L$.}
    \begin{tabular}{ccc}
        \toprule
        $\tau$ & $m$ & LT in Eq.~\eqref{eq:LRMPS}  \\
        \midrule\addlinespace
        $\tau\ll L(\ln L)^{11}$ & $\lesssim\dfrac{C'}{2}\qty(\dfrac{\tau}{L})^{1/11}$ & $2dL\chi^2 \ln L$\\ \addlinespace
        $L(\ln L)^{11}\ll\tau\ll Ld^{11L}$ & $\simeq\dfrac{C'}{2}\qty(\dfrac{\tau}{L})^{1/11}$ & $-m\ln \dfrac{d^L a^2}{m}~(\ll -L)$\\ \addlinespace
        $Ld^{11L} \ll \tau$ & $\simeq\dfrac{d^L a^2}{e}$ & $-\dfrac{d^L a^2}{e}$\\\addlinespace
        \bottomrule
    \end{tabular}
    \label{tab:LRMPS}
\end{table}

\paragraph{Discussion}
We have demonstrated a burst phenomenon starting from a low-entangled state for a nonintegrable mixed-field Ising chain, using average magnetization as an example.
By utilizing the DMRG algorithm, we have identified an MPS that creates a large burst at a designated time.
We have found that in the short-time regime when quantum scrambling sets in, a large burst accompanied by a slow or even negative entanglement growth arises from an MPS with a fixed bond dimension even for a large system.
This means that the burst from a low-entangled initial state can be realized by a quantum quench.
Conversely, at long times, the burst amplitude becomes small, which is consistent with our analytical bounds derived from local random circuits.
Our results suggest that starting from an appropriately chosen low-entangled state---rather than a typical one---allows the system to stay out of equilibrium for a prolonged time.
However, the state is eventually overwhelmed by quantum scrambling, which drives the system into equilibrium.

A burst can be regarded as a transient defiance against thermalization. Even if a local observable satisfies the ETH and initially takes its expectation value near the thermal equilibrium value, it can subsequently exhibit a significant deviation accompanied by an anomalous suppression of entanglement entropy.
We have shown that such an exceptional phenomenon can arise from low-entangled states, using a method systematically applicable to generic local Hamiltonians and observables.
Moreover, we expect that our variational method can be adapted to target other nonequilibrium trajectories---such as revival and oscillation---by modifying the cost function.
Such studies may offer insights into quantum thermalization beyond the infinite-time limit.

While we have investigated the dependence of the burst amplitude on the longitudinal magnetic field $h_z$, how a burst depends on the properties of the Hamiltonian is yet to be fully understood.
For instance, we expect that a burst decays faster when a system includes long-range interactions; however, it has yet to be confirmed for large systems due to computational cost.
Such a study may shed new light on how the system's properties such as locality of interactions affect thermalization processes.

The burst phenomenon offers a potential avenue for quantum metrology.
The burst amplitude of an intensive observable remains $\order{L^0}$ unlike thermal fluctuations, which are suppressed as $\order{L^{-1/2}}$.
This high signal-to-noise ratio could be exploited in quantum sensors and simulators: by comparing theoretical predictions with experimental results, one could estimate system parameters or benchmark the performance of quantum simulators.
While such applications were proposed in Ref.~\cite{Ermakov2021_complete}, we have improved their feasibility through simplified preparation of initial states.

\paragraph{Acknowledgments}
We thank Masaya Nakagawa, Shion Yamashika, and Atsushi Oyaizu for fruitful discussions, especially M. N. for giving the ``burst'' its name.
S.~Y. and A.~H. were supported by Forefront Physics and Mathematics Program to Drive Transformation (FoPM), a World-leading Innovative Graduate Study (WINGS) Program, the University of Tokyo. 
A.~H. was also supported by KAKENHI Grant No. JP25KJ0833 from the Japan Society for the Promotion of Science (JSPS) and JSR Fellowship, the University of Tokyo.
M.~U. was supported by KAKENHI Grant No. JP22H01152 from the JSPS. 
We gratefully acknowledge the support from the CREST program ``Quantum Frontiers'' (Grant No. JPMJCR23I1) by the Japan Science and Technology Agency (JST).
This work was supported by JST as part of Adopting Sustainable Partnerships for Innovative Research Ecosystem (ASPIRE), Grant No. JPMJAP25A1.

\paragraph{Data availability}
The code used to analyze the data is available from Ref.~\cite{github-repository}. 

\bibliography{main}

@ARTICLE{Neumann1929_beweis,
  title     = "{Beweis des Ergodensatzes und des $H$-Theorems in der neuen
               Mechanik}",
  author    = "Neumann, J v",
  journal   = "Z. Phys.",
  publisher = "Springer Science and Business Media LLC",
  volume    =  57,
  number    = "1-2",
  pages     = "30--70",
  year      =  1929,
  url       = "http://dx.doi.org/10.1007/bf01339852",
  doi       = "10.1007/bf01339852"
}

@ARTICLE{Bocchieri1957_quantum,
  title     = "{Quantum Recurrence Theorem}",
  author    = "Bocchieri, P and Loinger, A",
  journal   = "Phys. Rev.",
  publisher = "American Physical Society (APS)",
  volume    =  107,
  number    =  2,
  pages     = "337--338",
  year      =  1957,
  url       = "https://journals.aps.org/pr/abstract/10.1103/PhysRev.107.337",
  doi       = "10.1103/physrev.107.337"
}

@ARTICLE{Trotter1959_product,
  title     = "{On the product of semi-groups of operators}",
  author    = "Trotter, H F",
  journal   = "Proc. Am. Math. Soc.",
  publisher = "JSTOR",
  volume    =  10,
  number    =  4,
  pages     =  545,
  year      =  1959,
  url       = "http://dx.doi.org/10.2307/2033649",
  doi       = "10.2307/2033649"
}

@ARTICLE{Percival1961_periodicity,
  title     = "{Almost periodicity and the quantal \textit{H} theorem}",
  author    = "Percival, Ian C",
  journal   = "J. Math. Phys.",
  publisher = "AIP Publishing",
  volume    =  2,
  number    =  2,
  pages     = "235--239",
  year      =  1961,
  url       = "http://dx.doi.org/10.1063/1.1703705",
  doi       = "10.1063/1.1703705"
}

@ARTICLE{Suzuki1976_generalized,
  title     = "{Generalized Trotter's formula and systematic approximants of
               exponential operators and inner derivations with applications to
               many-body problems}",
  author    = "Suzuki, Masuo",
  journal   = "Commun. Math. Phys.",
  publisher = "Springer Science and Business Media LLC",
  volume    =  51,
  number    =  2,
  pages     = "183--190",
  year      =  1976,
  url       = "http://dx.doi.org/10.1007/BF01609348",
  doi       = "10.1007/bf01609348"
}

@ARTICLE{Peres1982_recurrence,
  title     = "{Recurrence phenomena in quantum dynamics}",
  author    = "Peres, Asher",
  journal   = "Phys. Rev. Lett.",
  publisher = "American Physical Society (APS)",
  volume    =  49,
  number    =  15,
  pages     = "1118--1118",
  year      =  1982,
  url       = "https://journals.aps.org/prl/pdf/10.1103/PhysRevLett.49.1118",
  doi       = "10.1103/physrevlett.49.1118"
}

@ARTICLE{Peres1984_stability,
  title     = "{Stability of quantum motion in chaotic and regular systems}",
  author    = "Peres, Asher",
  journal   = "Phys. Rev. A",
  publisher = "American Physical Society",
  volume    =  30,
  number    =  4,
  pages     =  1610,
  year      =  1984,
  url       = "http://dx.doi.org/10.1103/PhysRevA.30.1610",
  doi       = "10.1103/PhysRevA.30.1610"
}

@ARTICLE{Deutsch1991_quantum,
  title     = "{Quantum statistical mechanics in a closed system}",
  author    = "Deutsch, J M",
  journal   = "Phys. Rev. A",
  publisher = "American Physical Society (APS)",
  volume    =  43,
  number    =  4,
  pages     = "2046--2049",
  year      =  1991,
  url       = "https://journals.aps.org/pra/abstract/10.1103/PhysRevA.43.2046",
  doi       = "10.1103/physreva.43.2046"
}

@ARTICLE{Fannes1992_finitely,
  title     = "{Finitely correlated states on quantum spin chains}",
  author    = "Fannes, M and Nachtergaele, B and Werner, R F",
  journal   = "Commun. Math. Phys.",
  publisher = "Springer Science and Business Media LLC",
  volume    =  144,
  number    =  3,
  pages     = "443--490",
  year      =  1992,
  url       = "http://link.springer.com/10.1007/BF02099178",
  doi       = "10.1007/bf02099178"
}

@ARTICLE{White1992_density,
  title     = "{Density matrix formulation for quantum renormalization groups}",
  author    = "White, S R",
  journal   = "Phys. Rev. Lett.",
  publisher = "American Physical Society (APS)",
  volume    =  69,
  number    =  19,
  pages     = "2863--2866",
  year      =  1992,
  url       = "http://dx.doi.org/10.1103/PhysRevLett.69.2863",
  doi       = "10.1103/PhysRevLett.69.2863"
}

@ARTICLE{White1993_density-matrix,
  title     = "{Density-matrix algorithms for quantum renormalization groups}",
  author    = "White, S R",
  journal   = "Phys. Rev. B",
  publisher = "American Physical Society (APS)",
  volume    =  48,
  number    =  14,
  pages     = "10345--10356",
  year      =  1993,
  url       = "https://journals.aps.org/prb/abstract/10.1103/PhysRevB.48.10345",
  doi       = "10.1103/physrevb.48.10345"
}

@ARTICLE{Srednicki1994_chaos,
  title     = "{Chaos and quantum thermalization}",
  author    = "Srednicki, M",
  journal   = "Phys. Rev. E",
  publisher = "American Physical Society (APS)",
  volume    =  50,
  number    =  2,
  pages     = "888--901",
  year      =  1994,
  url       = "https://journals.aps.org/pre/abstract/10.1103/PhysRevE.50.888",
  doi       = "10.1103/physreve.50.888"
}

@ARTICLE{Berges2004_prethermalization,
  title     = "{Prethermalization}",
  author    = "Berges, J and Bors\'{a}nyi, Sz and Wetterich, C",
  journal   = "Phys. Rev. Lett.",
  publisher = "American Physical Society (APS)",
  volume    =  93,
  number    =  14,
  pages     =  142002,
  year      =  2004,
  url       = "http://dx.doi.org/10.1103/PhysRevLett.93.142002",
  doi       = "10.1103/PhysRevLett.93.142002"
}

@ARTICLE{Verstraete2004_matrix,
  title     = "{Matrix product density operators: simulation of
               finite-temperature and dissipative systems}",
  author    = "Verstraete, F and Garc\'{\i}a-Ripoll, J J and Cirac, J I",
  journal   = "Phys. Rev. Lett.",
  publisher = "American Physical Society (APS)",
  volume    =  93,
  number    =  20,
  pages     =  207204,
  year      =  2004,
  url       = "http://dx.doi.org/10.1103/PhysRevLett.93.207204",
  doi       = "10.1103/PhysRevLett.93.207204"
}

@ARTICLE{Calabrese2005_evolution,
  title     = "{Evolution of entanglement entropy in one-dimensional systems}",
  author    = "Calabrese, Pasquale and Cardy, John",
  journal   = "J. Stat. Mech.",
  publisher = "IOP Publishing",
  volume    =  2005,
  number    =  04,
  pages     = "P04010",
  year      =  2005,
  url       = "http://dx.doi.org/10.1088/1742-5468/2005/04/P04010",
  doi       = "10.1088/1742-5468/2005/04/P04010"
}

@ARTICLE{Schollwock2005_density-matrix,
  title     = "{The density-matrix renormalization group}",
  author    = "Schollw{\"{o}}ck, U",
  journal   = "Rev. Mod. Phys.",
  publisher = "American Physical Society",
  volume    =  77,
  number    =  1,
  pages     =  259,
  year      =  2005,
  url       = "http://dx.doi.org/10.1103/RevModPhys.77.259",
  doi       = "10.1103/RevModPhys.77.259"
}

@ARTICLE{Verstraete2006_matrix,
  title     = "{Matrix product states represent ground states faithfully}",
  author    = "Verstraete, F and Cirac, J I",
  journal   = "Phys. Rev. B",
  publisher = "American Physical Society",
  volume    =  73,
  number    =  9,
  pages     =  094423,
  year      =  2006,
  url       = "http://dx.doi.org/10.1103/PhysRevB.73.094423",
  doi       = "10.1103/PhysRevB.73.094423"
}

@ARTICLE{Kinoshita2006_quantum,
  title     = "{A quantum Newton's cradle}",
  author    = "Kinoshita, Toshiya and Wenger, Trevor and Weiss, David S",
  journal   = "Nature",
  publisher = "Springer Science and Business Media LLC",
  volume    =  440,
  number    =  7086,
  pages     = "900--903",
  year      =  2006,
  url       = "http://dx.doi.org/10.1038/nature04693",
  doi       = "10.1038/nature04693"
}

@ARTICLE{Basko2006_metal-insulator,
  title     = "{Metal--insulator transition in a weakly interacting
               many-electron system with localized single-particle states}",
  author    = "Basko, D M and Aleiner, I L and Altshuler, B L",
  journal   = "Ann. Phys.",
  publisher = "Elsevier BV",
  volume    =  321,
  number    =  5,
  pages     = "1126--1205",
  year      =  2006,
  url       = "http://dx.doi.org/10.1016/j.aop.2005.11.014",
  doi       = "10.1016/j.aop.2005.11.014"
}

@ARTICLE{Perez-Garcia2007_matrix,
  title     = "{Matrix product state representations}",
  author    = "Perez-Garcia, D and Verstraete, F and Wolf, M M and Cirac, J I",
  journal   = "Quantum Inf. Comput.",
  publisher = "Rinton Press",
  volume    =  7,
  number    = "5\&6",
  pages     = "401--430",
  year      =  2007,
  url       = "http://dx.doi.org/10.26421/qic7.5-6-1",
  doi       = "10.26421/qic7.5-6-1"
}

@ARTICLE{Hastings2007_area,
  title     = "{An area law for one-dimensional quantum systems}",
  author    = "Hastings, M B",
  journal   = "J. Stat. Mech.",
  publisher = "IOP Publishing",
  volume    =  2007,
  number    =  08,
  pages     = "P08024",
  year      =  2007,
  url       = "http://dx.doi.org/10.1088/1742-5468/2007/08/P08024",
  doi       = "10.1088/1742-5468/2007/08/P08024"
}

@ARTICLE{Rigol2008_thermalization,
  title     = "{Thermalization and its mechanism for generic isolated quantum
               systems}",
  author    = "Rigol, Marcos and Dunjko, Vanja and Olshanii, Maxim",
  journal   = "Nature",
  publisher = "Springer Science and Business Media LLC",
  volume    =  452,
  number    =  7189,
  pages     = "854--858",
  year      =  2008,
  url       = "http://dx.doi.org/10.1038/nature06838",
  doi       = "10.1038/nature06838"
}

@ARTICLE{Bartsch2009_dynamical,
  title     = "{Dynamical typicality of quantum expectation values}",
  author    = "Bartsch, Christian and Gemmer, Jochen",
  journal   = "Phys. Rev. Lett.",
  publisher = "American Physical Society (APS)",
  volume    =  102,
  number    =  11,
  pages     =  110403,
  year      =  2009,
  url       = "http://dx.doi.org/10.1103/physrevlett.102.110403",
  doi       = "10.1103/PhysRevLett.102.110403"
}

@ARTICLE{Low2009_large,
  title     = "{Large deviation bounds for $k$-designs}",
  author    = "Low, Richard A",
  journal   = "Proc. R. Soc. A",
  publisher = "The Royal Society",
  volume    =  465,
  number    =  2111,
  pages     = "3289--3308",
  year      =  2009,
  url       = "https://royalsocietypublishing.org/doi/10.1098/rspa.2009.0232",
  doi       = "10.1098/rspa.2009.0232"
}

@ARTICLE{Pirvu2010_matrix,
  title   = "{Matrix product operator representations}",
  author  = "Pirvu, B and Murg, V and Cirac, J I and Verstraete, F",
  journal = "New J. Phys.",
  volume  =  12,
  number  =  2,
  pages   =  025012,
  year    =  2010,
  url     = "http://dx.doi.org/10.1088/1367-2630/12/2/025012",
  doi     = "10.1088/1367-2630/12/2/025012"
}

@PHDTHESIS{Low2010_pseudo-randomness,
  title   = "{Pseudo-randomness and Learning in Quantum Computation}",
  author  = "Low, Richard A",
  year    =  2010,
  school  = "University of Bristol, UK"
}

@ARTICLE{Schollwock2011_density-matrix,
  title     = "{The density-matrix renormalization group in the age of matrix
               product states}",
  author    = "Schollw{\"{o}}ck, Ulrich",
  journal   = "Ann. Phys.",
  publisher = "Elsevier BV",
  volume    =  326,
  number    =  1,
  pages     = "96--192",
  year      =  2011,
  url       = "http://dx.doi.org/10.1016/j.aop.2010.09.012",
  doi       = "10.1016/j.aop.2010.09.012"
}

@ARTICLE{Banuls2011_strong,
  title     = "{Strong and weak thermalization of infinite nonintegrable quantum
               systems}",
  author    = "Ba\~{n}uls, M C and Cirac, J I and Hastings, M B",
  journal   = "Phys. Rev. Lett.",
  publisher = "American Physical Society (APS)",
  volume    =  106,
  number    =  5,
  pages     =  050405,
  year      =  2011,
  url       = "http://dx.doi.org/10.1103/physrevlett.106.050405",
  doi       = "10.1103/PhysRevLett.106.050405"
}

@ARTICLE{Poulin2011_quantum,
  title     = "{Quantum simulation of time-dependent Hamiltonians and the
               convenient illusion of Hilbert space}",
  author    = "Poulin, David and Qarry, Angie and Somma, Rolando and Verstraete,
               Frank",
  journal   = "Phys. Rev. Lett.",
  publisher = "American Physical Society (APS)",
  volume    =  106,
  number    =  17,
  pages     =  170501,
  year      =  2011,
  url       = "http://dx.doi.org/10.1103/physrevlett.106.170501",
  doi       = "10.1103/PhysRevLett.106.170501"
}

@ARTICLE{Polkovnikov2011_colloquium,
  title     = "{\textit{Colloquium}: Nonequilibrium dynamics of closed
               interacting quantum systems}",
  author    = "Polkovnikov, Anatoli and Sengupta, Krishnendu and Silva,
               Alessandro and Vengalattore, Mukund",
  journal   = "Rev. Mod. Phys.",
  publisher = "American Physical Society (APS)",
  volume    =  83,
  number    =  3,
  pages     = "863--883",
  year      =  2011,
  url       = "http://dx.doi.org/10.1103/RevModPhys.83.863",
  doi       = "10.1103/revmodphys.83.863"
}

@ARTICLE{Short2012_quantum,
  title     = "{Quantum equilibration in finite time}",
  author    = "Short, Anthony J and Farrelly, Terence C",
  journal   = "New J. Phys.",
  publisher = "IOP Publishing",
  volume    =  14,
  number    =  1,
  pages     =  013063,
  year      =  2012,
  url       = "https://iopscience.iop.org/article/10.1088/1367-2630/14/1/013063/meta",
  doi       = "10.1088/1367-2630/14/1/013063"
}

@ARTICLE{Trotzky2012_probing,
  title     = "{Probing the relaxation towards equilibrium in an isolated
               strongly correlated one-dimensional Bose gas}",
  author    = "Trotzky, S and Chen, Y-A and Flesch, A and McCulloch, I P and
               Schollw{\"{o}}ck, U and Eisert, J and Bloch, I",
  journal   = "Nat. Phys.",
  publisher = "Springer Science and Business Media LLC",
  volume    =  8,
  number    =  4,
  pages     = "325--330",
  year      =  2012,
  url       = "https://www.nature.com/articles/nphys2232",
  doi       = "10.1038/nphys2232"
}

@ARTICLE{Gring2012_relaxation,
  title     = "{Relaxation and prethermalization in an isolated quantum system}",
  author    = "Gring, M and Kuhnert, M and Langen, T and Kitagawa, T and Rauer,
               B and Schreitl, M and Mazets, I and Smith, D Adu and Demler, E
               and Schmiedmayer, J",
  journal   = "Science",
  publisher = "American Association for the Advancement of Science (AAAS)",
  volume    =  337,
  number    =  6100,
  pages     = "1318--1322",
  year      =  2012,
  url       = "http://dx.doi.org/10.1126/science.1224953",
  doi       = "10.1126/science.1224953"
}

@ARTICLE{Steinigeweg2013_eigenstate,
  title     = "{Eigenstate thermalization within isolated spin-chain systems}",
  author    = "Steinigeweg, R and Herbrych, J and Prelov\v{s}ek, P",
  journal   = "Phys. Rev. E",
  publisher = "American Physical Society (APS)",
  volume    =  87,
  number    =  1,
  pages     =  012118,
  year      =  2013,
  url       = "https://journals.aps.org/pre/abstract/10.1103/PhysRevE.87.012118",
  doi       = "10.1103/PhysRevE.87.012118"
}

@ARTICLE{Kim2013_ballistic,
  title     = "{Ballistic spreading of entanglement in a diffusive nonintegrable
               system}",
  author    = "Kim, Hyungwon and Huse, David A",
  journal   = "Phys. Rev. Lett.",
  publisher = "American Physical Society (APS)",
  volume    =  111,
  number    =  12,
  pages     =  127205,
  year      =  2013,
  url       = "http://dx.doi.org/10.1103/physrevlett.111.127205",
  doi       = "10.1103/PhysRevLett.111.127205"
}

@ARTICLE{Goldstein2013_time,
  title     = "{Time scales in the approach to equilibrium of macroscopic
               quantum systems}",
  author    = "Goldstein, Sheldon and Hara, Takashi and Tasaki, Hal",
  journal   = "Phys. Rev. Lett.",
  publisher = "American Physical Society (APS)",
  volume    =  111,
  number    =  14,
  pages     =  140401,
  year      =  2013,
  url       = "http://dx.doi.org/10.1103/PhysRevLett.111.140401",
  doi       = "10.1103/PhysRevLett.111.140401"
}

@ARTICLE{Langen2013_local,
  title     = "{Local emergence of thermal correlations in an isolated quantum
               many-body system}",
  author    = "Langen, T and Geiger, R and Kuhnert, M and Rauer, B and
               Schmiedmayer, J",
  journal   = "Nat. Phys.",
  publisher = "Springer Science and Business Media LLC",
  volume    =  9,
  number    =  10,
  pages     = "640--643",
  year      =  2013,
  url       = "http://dx.doi.org/10.1038/nphys2739",
  doi       = "10.1038/nphys2739"
}

@ARTICLE{Beugeling2014_finite-size,
  title     = "{Finite-size scaling of eigenstate thermalization}",
  author    = "Beugeling, W and Moessner, R and Haque, Masudul",
  journal   = "Phys. Rev. E",
  publisher = "American Physical Society (APS)",
  volume    =  89,
  number    =  4,
  pages     =  042112,
  year      =  2014,
  url       = "https://journals.aps.org/pre/pdf/10.1103/PhysRevE.89.042112",
  doi       = "10.1103/PhysRevE.89.042112"
}

@ARTICLE{Richerme2014_non-local,
  title     = "{Non-local propagation of correlations in quantum systems with
               long-range interactions}",
  author    = "Richerme, Philip and Gong, Zhe-Xuan and Lee, Aaron and Senko,
               Crystal and Smith, Jacob and Foss-Feig, Michael and Michalakis,
               Spyridon and Gorshkov, Alexey V and Monroe, Christopher",
  journal   = "Nature",
  publisher = "Springer Science and Business Media LLC",
  volume    =  511,
  number    =  7508,
  pages     = "198--201",
  year      =  2014,
  url       = "http://dx.doi.org/10.1038/nature13450",
  doi       = "10.1038/nature13450"
}

@ARTICLE{Kim2014_testing,
  title     = "{Testing whether all eigenstates obey the eigenstate
               thermalization hypothesis}",
  author    = "Kim, Hyungwon and Ikeda, Tatsuhiko N and Huse, David A",
  journal   = "Phys. Rev. E",
  publisher = "American Physical Society (APS)",
  volume    =  90,
  number    = "5-1",
  pages     =  052105,
  year      =  2014,
  url       = "https://journals.aps.org/pre/pdf/10.1103/PhysRevE.90.052105",
  doi       = "10.1103/PhysRevE.90.052105"
}

@ARTICLE{Beugeling2015_off-diagonal,
  title     = "{Off-diagonal matrix elements of local operators in many-body
               quantum systems}",
  author    = "Beugeling, Wouter and Moessner, Roderich and Haque, Masudul",
  journal   = "Phys. Rev. E",
  publisher = "American Physical Society (APS)",
  volume    =  91,
  number    =  1,
  pages     =  012144,
  year      =  2015,
  url       = "https://journals.aps.org/pre/abstract/10.1103/PhysRevE.91.012144",
  doi       = "10.1103/PhysRevE.91.012144"
}

@ARTICLE{Eisert2015_quantum,
  title     = "{Quantum many-body systems out of equilibrium}",
  author    = "Eisert, J and Friesdorf, M and Gogolin, C",
  journal   = "Nat. Phys.",
  publisher = "Springer Science and Business Media LLC",
  volume    =  11,
  number    =  2,
  pages     = "124--130",
  year      =  2015,
  url       = "http://dx.doi.org/10.1038/nphys3215",
  doi       = "10.1038/nphys3215"
}

@ARTICLE{Schreiber2015_observation,
  title     = "{Observation of many-body localization of interacting fermions in
               a quasirandom optical lattice}",
  author    = "Schreiber, Michael and Hodgman, Sean S and Bordia, Pranjal and
               L{\"{u}}schen, Henrik P and Fischer, Mark H and Vosk, Ronen and
               Altman, Ehud and Schneider, Ulrich and Bloch, Immanuel",
  journal   = "Science",
  publisher = "American Association for the Advancement of Science (AAAS)",
  volume    =  349,
  number    =  6250,
  pages     = "842--845",
  year      =  2015,
  url       = "https://www.science.org/doi/10.1126/science.aaa7432",
  doi       = "10.1126/science.aaa7432"
}

@ARTICLE{DAlessio2016_quantum,
  title     = "{From quantum chaos and eigenstate thermalization to statistical
               mechanics and thermodynamics}",
  author    = "D'Alessio, Luca and Kafri, Yariv and Polkovnikov, Anatoli and
               Rigol, Marcos",
  journal   = "Adv. Phys.",
  publisher = "Informa UK Limited",
  volume    =  65,
  number    =  3,
  pages     = "239--362",
  year      =  2016,
  url       = "http://dx.doi.org/10.1080/00018732.2016.1198134",
  doi       = "10.1080/00018732.2016.1198134"
}

@ARTICLE{Gogolin2016_equilibration,
  title     = "{Equilibration, thermalisation, and the emergence of statistical
               mechanics in closed quantum systems}",
  author    = "Gogolin, Christian and Eisert, Jens",
  journal   = "Rep. Prog. Phys.",
  publisher = "IOP Publishing",
  volume    =  79,
  number    =  5,
  pages     =  056001,
  year      =  2016,
  url       = "http://dx.doi.org/10.1088/0034-4885/79/5/056001",
  doi       = "10.1088/0034-4885/79/5/056001"
}

@ARTICLE{Imbrie2016_many-body,
  title     = "{On many-body localization for quantum spin chains}",
  author    = "Imbrie, John Z",
  journal   = "J. Stat. Phys.",
  publisher = "Springer Science and Business Media LLC",
  volume    =  163,
  number    =  5,
  pages     = "998--1048",
  year      =  2016,
  url       = "http://dx.doi.org/10.1007/s10955-016-1508-x",
  doi       = "10.1007/s10955-016-1508-x"
}

@ARTICLE{Kaufman2016_quantum,
  title     = "{Quantum thermalization through entanglement in an isolated
               many-body system}",
  author    = "Kaufman, Adam M and Tai, M Eric and Lukin, Alexander and Rispoli,
               Matthew and Schittko, Robert and Preiss, Philipp M and Greiner,
               Markus",
  journal   = "Science",
  publisher = "American Association for the Advancement of Science (AAAS)",
  volume    =  353,
  number    =  6301,
  pages     = "794--800",
  year      =  2016,
  url       = "http://dx.doi.org/10.1126/science.aaf6725",
  doi       = "10.1126/science.aaf6725"
}

@ARTICLE{Brandao2016_local,
  title     = "{Local random quantum circuits are approximate
               polynomial-designs}",
  author    = "Brand\~{a}o, Fernando G S L and Harrow, Aram W and Horodecki,
               Micha\l{}",
  journal   = "Commun. Math. Phys.",
  publisher = "Springer Science and Business Media LLC",
  volume    =  346,
  number    =  2,
  pages     = "397--434",
  year      =  2016,
  url       = "http://dx.doi.org/10.1007/s00220-016-2706-8",
  doi       = "10.1007/s00220-016-2706-8"
}

@ARTICLE{Smith2016_many-body,
  title     = "{Many-body localization in a quantum simulator with programmable
               random disorder}",
  author    = "Smith, J and Lee, A and Richerme, P and Neyenhuis, B and Hess, P
               W and Hauke, P and Heyl, M and Huse, D A and Monroe, C",
  journal   = "Nat. Phys.",
  publisher = "Springer Science and Business Media LLC",
  volume    =  12,
  number    =  10,
  pages     = "907--911",
  year      =  2016,
  url       = "http://dx.doi.org/10.1038/nphys3783",
  doi       = "10.1038/nphys3783"
}

@ARTICLE{Clos2016_time-resolved,
  title     = "{Time-resolved observation of thermalization in an isolated
               quantum system}",
  author    = "Clos, Govinda and Porras, Diego and Warring, Ulrich and Schaetz,
               Tobias",
  journal   = "Phys. Rev. Lett.",
  publisher = "American Physical Society",
  volume    =  117,
  number    =  17,
  pages     =  170401,
  year      =  2016,
  url       = "http://dx.doi.org/10.1103/PhysRevLett.117.170401",
  doi       = "10.1103/PhysRevLett.117.170401"
}

@ARTICLE{Neill2016_ergodic,
  title     = "{Ergodic dynamics and thermalization in an isolated quantum
               system}",
  author    = "Neill, C and Roushan, P and Fang, M and Chen, Y and Kolodrubetz,
               M and Chen, Z and Megrant, A and Barends, R and Campbell, B and
               Chiaro, B and Dunsworth, A and Jeffrey, E and Kelly, J and Mutus,
               J and O'Malley, P J J and Quintana, C and Sank, D and
               Vainsencher, A and Wenner, J and White, T C and Polkovnikov, A
               and Martinis, J M",
  journal   = "Nat. Phys.",
  publisher = "Springer Science and Business Media LLC",
  volume    =  12,
  number    =  11,
  pages     = "1037--1041",
  year      =  2016,
  url       = "http://dx.doi.org/10.1038/nphys3830",
  doi       = "10.1038/nphys3830"
}

@ARTICLE{Zhou2017_operator,
  title     = "{Operator entanglement entropy of the time evolution operator in
               chaotic systems}",
  author    = "Zhou, Tianci and Luitz, David J",
  journal   = "Phys. Rev. B",
  publisher = "American Physical Society",
  volume    =  95,
  number    =  9,
  pages     =  094206,
  year      =  2017,
  url       = "http://dx.doi.org/10.1103/physrevb.95.094206",
  doi       = "10.1103/PhysRevB.95.094206"
}

@ARTICLE{Mondaini2017_eigenstate,
  title   = "{Eigenstate thermalization in the two-dimensional transverse field
             Ising model. II. Off-diagonal matrix elements of observables}",
  author  = "Mondaini, Rubem and Rigol, Marcos",
  journal = "Phys. Rev. E",
  volume  =  96,
  number  = "1-1",
  pages   =  012157,
  year    =  2017,
  url     = "http://dx.doi.org/10.1103/physreve.96.012157",
  doi     = "10.1103/PhysRevE.96.012157"
}

@ARTICLE{Bernien2017_probing,
  title     = "{Probing many-body dynamics on a 51-atom quantum simulator}",
  author    = "Bernien, H and Schwartz, S and Keesling, A and Levine, H and
               Omran, Ahmed and Pichler, H and Choi, Soonwon and Zibrov, A and
               Endres, M and Greiner, M and Vuleti\'{c}, V and Lukin, M",
  journal   = "Nature",
  publisher = "nature.com",
  volume    =  551,
  pages     = "579--584",
  year      =  2017,
  url       = "https://idp.nature.com/authorize/casa?redirect_uri=https://www.nature.com/articles/nature24622&casa_token=cKjy38URa5gAAAAA:D2rHbh9sTcpeT87yReqYXJOj1TzA5DP0cfX-hJrns5MNLIlaConES9XOhoj9mZbUqISNNhatCfxpXH4RBKs",
  doi       = "10.1038/nature24622"
}

@ARTICLE{Shiraishi2017_systematic,
  title     = "{Systematic construction of counterexamples to the eigenstate
               thermalization hypothesis}",
  author    = "Shiraishi, Naoto and Mori, Takashi",
  journal   = "Phys. Rev. Lett.",
  publisher = "Amer Physical Soc",
  volume    =  119,
  number    =  3,
  pages     =  030601,
  year      =  2017,
  url       = "http://dx.doi.org/10.1103/PhysRevLett.119.030601",
  doi       = "10.1103/PhysRevLett.119.030601"
}

@ARTICLE{Turner2018_weak,
  title     = "{Weak ergodicity breaking from quantum many-body scars}",
  author    = "Turner, C J and Michailidis, A A and Abanin, D A and Serbyn, M
               and Papi\'{c}, Z",
  journal   = "Nat. Phys.",
  publisher = "Springer Science and Business Media LLC",
  volume    =  14,
  number    =  7,
  pages     = "745--749",
  year      =  2018,
  url       = "http://dx.doi.org/10.1038/s41567-018-0137-5",
  doi       = "10.1038/s41567-018-0137-5"
}

@ARTICLE{Neyenhuis2017_observation,
  title     = "{Observation of prethermalization in long-range interacting spin
               chains}",
  author    = "Neyenhuis, Brian and Zhang, Jiehang and Hess, Paul W and Smith,
               Jacob and Lee, Aaron C and Richerme, Phil and Gong, Zhe-Xuan and
               Gorshkov, Alexey V and Monroe, Christopher",
  journal   = "Sci. Adv.",
  publisher = "American Association for the Advancement of Science (AAAS)",
  volume    =  3,
  number    =  8,
  pages     = "e1700672",
  year      =  2017,
  url       = "http://dx.doi.org/10.1126/sciadv.1700672",
  doi       = "10.1126/sciadv.1700672"
}

@ARTICLE{Garrison2018_single,
  title     = "{Does a single eigenstate encode the full Hamiltonian?}",
  author    = "Garrison, James R and Grover, Tarun",
  journal   = "Phys. Rev. X",
  publisher = "American Physical Society (APS)",
  volume    =  8,
  number    =  2,
  pages     =  021026,
  year      =  2018,
  url       = "https://journals.aps.org/prx/abstract/10.1103/PhysRevX.8.021026",
  doi       = "10.1103/physrevx.8.021026"
}

@ARTICLE{Mori2018_thermalization,
  title     = "{Thermalization and prethermalization in isolated quantum
               systems: a theoretical overview}",
  author    = "Mori, Takashi and Ikeda, Tatsuhiko N and Kaminishi, Eriko and
               Ueda, Masahito",
  journal   = "J. Phys. B: At. Mol. Opt. Phys.",
  publisher = "IOP Publishing",
  volume    =  51,
  number    =  11,
  pages     =  112001,
  year      =  2018,
  url       = "https://iopscience.iop.org/article/10.1088/1361-6455/aabcdf/meta",
  doi       = "10.1088/1361-6455/aabcdf"
}

@ARTICLE{Dymarsky2019_mechanism,
  title     = "{Mechanism of macroscopic equilibration of isolated quantum
               systems}",
  author    = "Dymarsky, A",
  journal   = "Phys. Rev. B",
  publisher = "APS",
  volume    =  99,
  pages     =  224302,
  year      =  2019,
  url       = "https://journals.aps.org/prb/abstract/10.1103/PhysRevB.99.224302",
  doi       = "10.1103/PhysRevB.99.224302"
}

@ARTICLE{Knipschild2020_modern,
  title     = "{Modern concepts of quantum equilibration do not rule out strange
               relaxation dynamics}",
  author    = "Knipschild, L and Gemmer, J",
  journal   = "Phys. Rev. E",
  publisher = "APS",
  volume    =  101,
  pages     =  062205,
  year      =  2020,
  url       = "https://journals.aps.org/pre/abstract/10.1103/PhysRevE.101.062205",
  doi       = "10.1103/physreve.101.062205"
}

@ARTICLE{Sala2020_ergodicity,
  title     = "{Ergodicity breaking arising from Hilbert space fragmentation in
               dipole-conserving Hamiltonians}",
  author    = "Sala, Pablo and Rakovszky, Tibor and Verresen, Ruben and Knap,
               Michael and Pollmann, Frank",
  journal   = "Phys. Rev. X",
  publisher = "American Physical Society (APS)",
  volume    =  10,
  number    =  1,
  pages     =  011047,
  year      =  2020,
  url       = "https://journals.aps.org/prx/abstract/10.1103/PhysRevX.10.011047",
  doi       = "10.1103/physrevx.10.011047"
}

@ARTICLE{Ran2020_encoding,
  title     = "{Encoding of matrix product states into quantum circuits of one-
               and two-qubit gates}",
  author    = "Ran, Shi-Ju",
  journal   = "Phys. Rev. A",
  publisher = "American Physical Society (APS)",
  volume    =  101,
  number    =  3,
  pages     =  032310,
  year      =  2020,
  url       = "https://journals.aps.org/pra/pdf/10.1103/PhysRevA.101.032310",
  doi       = "10.1103/physreva.101.032310"
}

@ARTICLE{Khemani2020_localization,
  title     = "{Localization from Hilbert space shattering: From theory to
               physical realizations}",
  author    = "Khemani, Vedika and Hermele, Michael and Nandkishore, Rahul",
  journal   = "Phys. Rev. B",
  publisher = "American Physical Society (APS)",
  volume    =  101,
  number    =  17,
  pages     =  174204,
  year      =  2020,
  url       = "http://dx.doi.org/10.1103/PhysRevB.101.174204",
  doi       = "10.1103/physrevb.101.174204"
}

@ARTICLE{Sugimoto2021_test,
  title     = "{Test of the eigenstate thermalization hypothesis based on local
               random matrix theory}",
  author    = "Sugimoto, Shoki and Hamazaki, Ryusuke and Ueda, Masahito",
  journal   = "Phys. Rev. Lett.",
  publisher = "American Physical Society (APS)",
  volume    =  126,
  number    =  12,
  pages     =  120602,
  year      =  2021,
  url       = "https://journals.aps.org/prl/pdf/10.1103/PhysRevLett.126.120602",
  doi       = "10.1103/PhysRevLett.126.120602"
}

@ARTICLE{Cao2021_statistical,
  title   = "{A statistical mechanism for operator growth}",
  author  = "Cao, Xiangyu",
  journal = "J. Phys. A: Math. Theor.",
  volume  =  54,
  number  =  14,
  pages   =  144001,
  year    =  2021,
  url     = "http://dx.doi.org/10.1088/1751-8121/abe77c",
  doi     = "10.1088/1751-8121/abe77c"
}

@ARTICLE{Blok2021_quantum,
  title     = "{Quantum information scrambling on a superconducting qutrit
               processor}",
  author    = "Blok, M S and Ramasesh, V V and Schuster, T and O'Brien, K and
               Kreikebaum, J M and Dahlen, D and Morvan, A and Yoshida, B and
               Yao, N Y and Siddiqi, I",
  journal   = "Phys. Rev. X",
  publisher = "American Physical Society (APS)",
  volume    =  11,
  number    =  2,
  pages     =  021010,
  year      =  2021,
  url       = "http://dx.doi.org/10.1103/PhysRevX.11.021010",
  doi       = "10.1103/physrevx.11.021010"
}

@ARTICLE{Chen2021_observation,
  title     = "{Observation of strong and weak thermalization in a
               superconducting quantum processor}",
  author    = "Chen, Fusheng and Sun, Zheng-Hang and Gong, Ming and Zhu,
               Qingling and Zhang, Yu-Ran and Wu, Yulin and Ye, Yangsen and Zha,
               Chen and Li, Shaowei and Guo, Shaojun and Qian, Haoran and Huang,
               He-Liang and Yu, Jiale and Deng, Hui and Rong, Hao and Lin, Jin
               and Xu, Yu and Sun, Lihua and Guo, Cheng and Li, Na and Liang,
               Futian and Peng, Cheng-Zhi and Fan, Heng and Zhu, Xiaobo and Pan,
               Jian-Wei",
  journal   = "Phys. Rev. Lett.",
  publisher = "American Physical Society (APS)",
  volume    =  127,
  number    =  2,
  pages     =  020602,
  year      =  2021,
  url       = "http://dx.doi.org/10.1103/PhysRevLett.127.020602",
  doi       = "10.1103/PhysRevLett.127.020602"
}

@ARTICLE{Scherg2021_observing,
  title     = "{Observing non-ergodicity due to kinetic constraints in tilted
               Fermi-Hubbard chains}",
  author    = "Scherg, Sebastian and Kohlert, Thomas and Sala, Pablo and
               Pollmann, Frank and Hebbe Madhusudhana, Bharath and Bloch,
               Immanuel and Aidelsburger, Monika",
  journal   = "Nat. Commun.",
  publisher = "Springer Science and Business Media LLC",
  volume    =  12,
  number    =  1,
  pages     =  4490,
  year      =  2021,
  url       = "http://dx.doi.org/10.1038/s41467-021-24726-0",
  doi       = "10.1038/s41467-021-24726-0"
}

@ARTICLE{Alhambra2021_locally,
  title     = "{Locally accurate tensor networks for thermal states and time
               evolution}",
  author    = "Alhambra, {\'{A}}lvaro M and Cirac, J Ignacio",
  journal   = "PRX Quantum",
  publisher = "American Physical Society (APS)",
  volume    =  2,
  number    =  4,
  pages     =  040331,
  year      =  2021,
  url       = "https://journals.aps.org/prxquantum/pdf/10.1103/PRXQuantum.2.040331",
  doi       = "10.1103/prxquantum.2.040331"
}

@ARTICLE{Ermakov2021_complete,
  title     = "{Almost complete revivals in quantum many-body systems}",
  author    = "Ermakov, Igor and Fine, Boris V",
  journal   = "Phys. Rev. A",
  publisher = "American Physical Society (APS)",
  volume    =  104,
  number    =  5,
  pages     = "L050202",
  year      =  2021,
  url       = "https://link.aps.org/doi/10.1103/PhysRevA.104.L050202",
  doi       = "10.1103/physreva.104.l050202"
}

@ARTICLE{Fishman2022_codebase,
  title     = "{Codebase release 0.3 for {ITensor}}",
  author    = "Fishman, Matthew and White, Steven R and Stoudenmire, E Miles",
  journal   = "SciPost Phys. Codebases",
  publisher = "SciPost",
  volume    = "{\normalfont 4-r0.3}",
  pages     = "",
  year      =  2022,
  url       = "http://dx.doi.org/10.21468/SciPostPhysCodeb.4-r0.3",
  doi       = "10.21468/SciPostPhysCodeb.4-r0.3"
}

@ARTICLE{Zhu2022_observation,
  title     = "{Observation of thermalization and information scrambling in a
               superconducting quantum processor}",
  author    = "Zhu, Qingling and Sun, Zheng-Hang and Gong, Ming and Chen,
               Fusheng and Zhang, Yu-Ran and Wu, Yulin and Ye, Yangsen and Zha,
               Chen and Li, Shaowei and Guo, Shaojun and Qian, Haoran and Huang,
               He-Liang and Yu, Jiale and Deng, Hui and Rong, Hao and Lin, Jin
               and Xu, Yu and Sun, Lihua and Guo, Cheng and Li, Na and Liang,
               Futian and Peng, Cheng-Zhi and Fan, Heng and Zhu, Xiaobo and Pan,
               Jian-Wei",
  journal   = "Phys. Rev. Lett.",
  publisher = "American Physical Society",
  volume    =  128,
  number    =  16,
  pages     =  160502,
  year      =  2022,
  url       = "http://dx.doi.org/10.1103/PhysRevLett.128.160502",
  doi       = "10.1103/PhysRevLett.128.160502"
}

@ARTICLE{Sompet2022_realizing,
  title     = "{Realizing the symmetry-protected Haldane phase in Fermi-Hubbard
               ladders}",
  author    = "Sompet, Pimonpan and Hirthe, Sarah and Bourgund, Dominik and
               Chalopin, Thomas and Bibo, Julian and Koepsell, Joannis and
               Bojovi\'{c}, Petar and Verresen, Ruben and Pollmann, Frank and
               Salomon, Guillaume and Gross, Christian and Hilker, Timon A and
               Bloch, Immanuel",
  journal   = "Nature",
  publisher = "Springer Science and Business Media LLC",
  volume    =  606,
  number    =  7914,
  pages     = "484--488",
  year      =  2022,
  url       = "http://dx.doi.org/10.1038/s41586-022-04688-z",
  doi       = "10.1038/s41586-022-04688-z"
}

@ARTICLE{Fishman2022_itensor,
  title     = "{The ITensor software library for tensor network calculations}",
  author    = "Fishman, Matthew and White, Steven R and Stoudenmire, Edwin M",
  journal   = "SciPost Phys. Codebases",
  publisher = "Stichting SciPost",
  number    =  4,
  volume    = "{\normalfont 4}",
  pages     = "",
  year      =  2022,
  url       = "https://scipost.org/10.21468/SciPostPhysCodeb.4",
  doi       = "10.21468/scipostphyscodeb.4"
}

@ARTICLE{Zhang2023_many-body,
  title     = "{Many-body Hilbert space scarring on a superconducting processor}",
  author    = "Zhang, Pengfei and Dong, Hang and Gao, Yu and Zhao, Liangtian and
               Hao, Jie and Desaules, Jean-Yves and Guo, Qiujiang and Chen,
               Jiachen and Deng, Jinfeng and Liu, Bobo and Ren, Wenhui and Yao,
               Yunyan and Zhang, Xu and Xu, Shibo and Wang, Ke and Jin, Feitong
               and Zhu, Xuhao and Zhang, Bing and Li, Hekang and Song, Chao and
               Wang, Zhen and Liu, Fangli and Papi\'{c}, Zlatko and Ying, Lei
               and Wang, H and Lai, Ying-Cheng",
  journal   = "Nat. Phys.",
  publisher = "Springer Science and Business Media LLC",
  volume    =  19,
  number    =  1,
  pages     = "120--125",
  year      =  2023,
  url       = "http://dx.doi.org/10.1038/s41567-022-01784-9",
  doi       = "10.1038/s41567-022-01784-9"
}

@ARTICLE{Rudolph2024_decomposition,
  title   = "{Decomposition of matrix product states into shallow quantum
             circuits}",
  author  = "Rudolph, Manuel S and Chen, Jing and Miller, Jacob and Acharya,
             Atithi and Perdomo-Ortiz, Alejandro",
  journal = "Quantum Sci. Technol.",
  volume  =  9,
  number  =  1,
  pages   =  015012,
  year    =  2024,
  url     = "http://dx.doi.org/10.1088/2058-9565/ad04e6",
  doi     = "10.1088/2058-9565/ad04e6"
}

@ARTICLE{Chiba2024_proof,
  title     = "{Proof of absence of local conserved quantities in the
               mixed-field Ising chain}",
  author    = "Chiba, Yuuya",
  journal   = "Phys. Rev. B",
  publisher = "American Physical Society (APS)",
  volume    =  109,
  number    =  3,
  pages     =  035123,
  year      =  2024,
  url       = "https://journals.aps.org/prb/pdf/10.1103/PhysRevB.109.035123",
  doi       = "10.1103/physrevb.109.035123"
}

@ARTICLE{Iaconis2024_quantum,
  title     = "{Quantum state preparation of normal distributions using matrix
               product states}",
  author    = "Iaconis, Jason and Johri, Sonika and Zhu, Elton Yechao",
  journal   = "npj Quantum Inf.",
  publisher = "Springer Science and Business Media LLC",
  volume    =  10,
  number    =  1,
  pages     =  15,
  year      =  2024,
  url       = "https://www.nature.com/articles/s41534-024-00805-0",
  doi       = "10.1038/s41534-024-00805-0"
}

@INPROCEEDINGS{Aubrun2026_optimal,
  author    = {Aubrun, Guillaume and Jenkinson, Justin
               and Szarek, Stanislaw J.},
  title     = {Optimal Constants in Concentration Inequalities on the Sphere
               and in the {Gauss} Space},
  booktitle = {High Dimensional Probability X},
  editor    = {Adamczak, Rados{\l}aw and Gozlan, Nathael
               and Lounici, Karim and Madiman, Mokshay
               and Merlev{\`e}de, Florence and Werner, Elisabeth},
  series    = {Progress in Probability},
  volume    = {82},
  pages     = {1--36},
  publisher = {Birkh{\"a}user},
  address   = {Cham},
  year      = {2026},
  doi       = {10.1007/978-3-032-06057-0_1}
}

@ARTICLE{Shirakawa2024_automatic,
  title     = "{Automatic quantum circuit encoding of a given arbitrary quantum
               state}",
  author    = "Shirakawa, Tomonori and Ueda, Hiroshi and Yunoki, Seiji",
  journal   = "Phys. Rev. Res.",
  publisher = "American Physical Society (APS)",
  volume    =  6,
  number    =  4,
  pages     =  043008,
  year      =  2024,
  url       = "https://journals.aps.org/prresearch/pdf/10.1103/PhysRevResearch.6.043008",
  doi       = "10.1103/physrevresearch.6.043008"
}

@ARTICLE{Li2024_designs,
  title     = "{Designs from local random quantum circuits with SU($d$)
               symmetry}",
  author    = "Li, Zimu and Zheng, Han and Liu, Junyu and Jiang, Liang and Liu,
               Zi-Wen",
  journal   = "PRX Quantum",
  publisher = "American Physical Society (APS)",
  volume    =  5,
  number    =  4,
  pages     =  040349,
  year      =  2024,
  url       = "https://journals.aps.org/prxquantum/pdf/10.1103/PRXQuantum.5.040349",
  doi       = "10.1103/prxquantum.5.040349"
}

@ARTICLE{Wang2025_exploring,
  title     = "{Exploring Hilbert-space fragmentation on a superconducting
               processor}",
  author    = "Wang, Yong-Yi and Shi, Yun-Hao and Sun, Zheng-Hang and Chen,
               Chi-Tong and Wang, Zheng-An and Zhao, Kui and Liu, Hao-Tian and
               Ma, Wei-Guo and Wang, Ziting and Li, Hao and Zhang, Jia-Chi and
               Liu, Yu and Deng, Cheng-Lin and Li, Tian-Ming and He, Yang and
               Liu, Zheng-He and Peng, Zhen-Yu and Song, Xiaohui and Xue,
               Guangming and Yu, Haifeng and Huang, Kaixuan and Xiang,
               Zhongcheng and Zheng, Dongning and Xu, Kai and Fan, Heng",
  journal   = "PRX Quantum",
  publisher = "American Physical Society (APS)",
  volume    =  6,
  number    =  1,
  pages     =  010325,
  year      =  2025,
  url       = "https://journals.aps.org/prxquantum/pdf/10.1103/PRXQuantum.6.010325",
  doi       = "10.1103/prxquantum.6.010325"
}

@ARTICLE{Hearth2025_unitary,
  title     = "{Unitary $k$-designs from random number-conserving quantum
               circuits}",
  author    = "Hearth, Sumner N and Flynn, Michael O and Chandran, Anushya and
               Laumann, Chris R",
  journal   = "Phys. Rev. X",
  publisher = "American Physical Society (APS)",
  volume    =  15,
  number    =  2,
  pages     =  021022,
  year      =  2025,
  url       = "https://journals.aps.org/prx/abstract/10.1103/PhysRevX.15.021022",
  doi       = "10.1103/physrevx.15.021022"
}

@ARTICLE{Mitsuhashi2025_unitary,
  title     = "{Unitary designs of symmetric local random circuits}",
  author    = "Mitsuhashi, Yosuke and Suzuki, Ryotaro and Soejima, Tomohiro and
               Yoshioka, Nobuyuki",
  journal   = "Phys. Rev. Lett.",
  publisher = "American Physical Society (APS)",
  volume    =  134,
  number    =  18,
  pages     =  180404,
  year      =  2025,
  url       = "http://dx.doi.org/10.1103/PhysRevLett.134.180404",
  doi       = "10.1103/PhysRevLett.134.180404"
}

@ARTICLE{Fine2014_absence,
  title     = "{Absence of exponential sensitivity to small perturbations in
               nonintegrable systems of spins 1/2}",
  author    = "Fine, B V and Elsayed, T A and Kropf, C M and de Wijn, A S",
  journal   = "Phys. Rev. E",
  publisher = "American Physical Society (APS)",
  volume    =  89,
  number    =  1,
  pages     =  012923,
  year      =  2014,
  url       = "https://journals.aps.org/pre/abstract/10.1103/PhysRevE.89.012923",
  doi       = "10.1103/PhysRevE.89.012923"
}

@ARTICLE{Schmitt2018_irreversible,
  title     = "{Irreversible dynamics in quantum many-body systems}",
  author    = "Schmitt, Markus and Kehrein, Stefan",
  journal   = "Phys. Rev. B",
  publisher = "American Physical Society (APS)",
  volume    =  98,
  number    =  18,
  pages     =  180301,
  year      =  2018,
  url       = "https://journals.aps.org/prb/pdf/10.1103/PhysRevB.98.180301",
  doi       = "10.1103/physrevb.98.180301"
}

@ARTICLE{Dabelow2020_persistent,
  title     = "{Persistent many-body quantum echoes}",
  author    = "Dabelow, Lennart and Reimann, Peter",
  journal   = "Phys. Rev. Res.",
  publisher = "American Physical Society (APS)",
  volume    =  2,
  number    =  2,
  pages     =  023216,
  year      =  2020,
  url       = "http://dx.doi.org/10.1103/PhysRevResearch.2.023216",
  doi       = "10.1103/physrevresearch.2.023216"
}

@ARTICLE{Wilming2019_entanglement-ergodic,
  title     = "{Entanglement-ergodic quantum systems equilibrate exponentially
               well}",
  author    = "Wilming, H and Goihl, M and Roth, I and Eisert, J",
  journal   = "Phys. Rev. Lett.",
  publisher = "American Physical Society (APS)",
  volume    =  123,
  number    =  20,
  pages     =  200604,
  year      =  2019,
  url       = "http://dx.doi.org/10.1103/PhysRevLett.123.200604",
  doi       = "10.1103/PhysRevLett.123.200604"
}

@ARTICLE{Bertoni2025_typical,
  title     = "{Typical thermalization of low-entanglement states}",
  author    = "Bertoni, Christian and Wassner, Clara and Guarnieri, Giacomo and
               Eisert, Jens",
  journal   = "Commun. Phys.",
  publisher = "Springer Science and Business Media LLC",
  volume    =  8,
  number    =  1,
  pages     =  301,
  year      =  2025,
  url       = "https://www.nature.com/articles/s42005-025-02161-7",
  doi       = "10.1038/s42005-025-02161-7"
}

@ARTICLE{Pilatowsky-Cameo2026_quantum,
  title     = "{Quantum thermalization must occur in translation-invariant
               systems at high temperature}",
  author    = "Pilatowsky-Cameo, Sa\'{u}l and Choi, Soonwon",
  journal   = "Nat. Commun.",
  publisher = "Springer Science and Business Media LLC",
  volume    =  17,
  number    =  1,
  pages     =  75,
  year      =  2026,
  url       = "http://dx.doi.org/10.1038/s41467-025-66777-7",
  doi       = "10.1038/s41467-025-66777-7"
}

@ARTICLE{Brown2018_second,
  title     = "{Second law of quantum complexity}",
  author    = "Brown, Adam R and Susskind, Leonard",
  journal   = "Phys. Rev. D",
  publisher = "American Physical Society",
  volume    =  97,
  number    =  8,
  pages     =  086015,
  year      =  2018,
  url       = "http://dx.doi.org/10.1103/PhysRevD.97.086015",
  doi       = "10.1103/PhysRevD.97.086015"
}

@ARTICLE{Orus2008_infinite,
  title     = "{Infinite time-evolving block decimation algorithm beyond unitary
               evolution}",
  author    = "Or\'{u}s, R and Vidal, G",
  journal   = "Phys. Rev. B",
  publisher = "American Physical Society (APS)",
  volume    =  78,
  number    =  15,
  pages     =  155117,
  year      =  2008,
  url       = "https://journals.aps.org/prb/pdf/10.1103/PhysRevB.78.155117",
  doi       = "10.1103/physrevb.78.155117"
}

@ARTICLE{Vidal2007_classical,
  title     = "{Classical simulation of infinite-size quantum lattice systems in
               one spatial dimension}",
  author    = "Vidal, G",
  journal   = "Phys. Rev. Lett.",
  publisher = "American Physical Society (APS)",
  volume    =  98,
  number    =  7,
  pages     =  070201,
  year      =  2007,
  url       = "http://dx.doi.org/10.1103/physrevlett.98.070201",
  doi       = "10.1103/PhysRevLett.98.070201"
}

@misc{MPO-truncation,
  note = {We perform Step~\ref{m2:heisenberg} of our method with the maximum bond dimension of $2^{11}$. For $L=10$, $20$, $30$, and $40$, the truncation error defined by $\epsilon_O(\tau)\coloneqq 1-(\|O(\tau)\|_F/\|O\|_F)^2$ under the Hamiltonian~\eqref{eq:Ising} satisfies $\epsilon_O(20)<6\times 10^{-4}$ and $\epsilon_O(30)<10^{-2}$ for $O=M^z, M^y$ and $(J_z, h_x, h_z) = (1, 0.9045, 0.8090)$.}
}

@misc{method-comparison,
  note = {The burst amplitudes obtained by the time-reversal-and-truncation method, the DMRG method ($\lambda_L=0$), and the DMRG method ($\lambda_L=72/L^2$) with $L=40$, $\tau=20$, and $\chi=10$ under the Hamiltonian~\eqref{eq:Ising} are 0.337, 0.396, and 0.330 for $O=M^z$, and 0.222, 0.371, and 0.305 for $O=M^y$, respectively.}
}

@misc{github-repository,
    howpublished = {\url{https://github.com/shozoyamada/many-body_burst}}
}

@misc{MPS-injectivity,
  note = {Strictly speaking, an MPS needs to be injective for it to be a {\it unique} ground state of a local Hamiltonian~\cite{Perez-Garcia2007_matrix}. The injectivity of the initial state is supported by its finite correlation length demonstrated in End Matter.}
}

@misc{comment-zero,
  note = {For $m=0$, we use the conventions $x^x|_{x= 0}=1$ and $x\ln x|_{x=0}=0$.}
}

\clearpage
\appendix
\section*{End Matter}

\begin{figure}[!hb]
    \centering
    \begin{minipage}{\linewidth}
        \centering
        \begin{tikzpicture}
            \node[anchor=south west, inner sep=0] (image) at (0,0) {\includegraphics[width=0.97\linewidth]{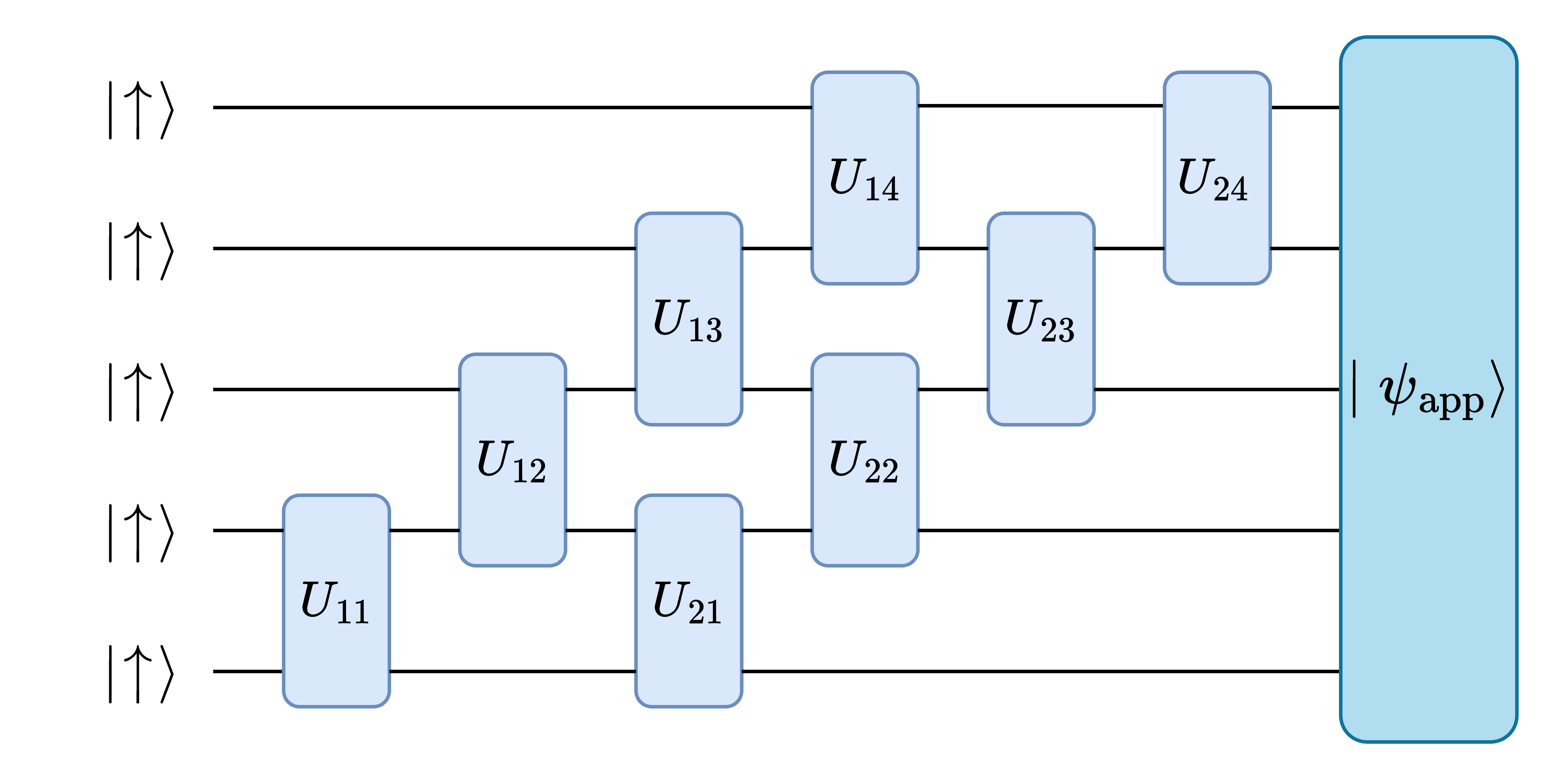}};
            
            \begin{scope}[x={(image.south east)},y={(image.north west)}]
            \node[anchor=north west] at (-0.025, 0.97) {(a)};
            \end{scope}
        \end{tikzpicture}
    \end{minipage}
    \begin{minipage}{\linewidth}
        \centering
        \begin{tikzpicture}
            \node[anchor=south west, inner sep=0] (image) at (0,0) {\includegraphics[width=\linewidth]{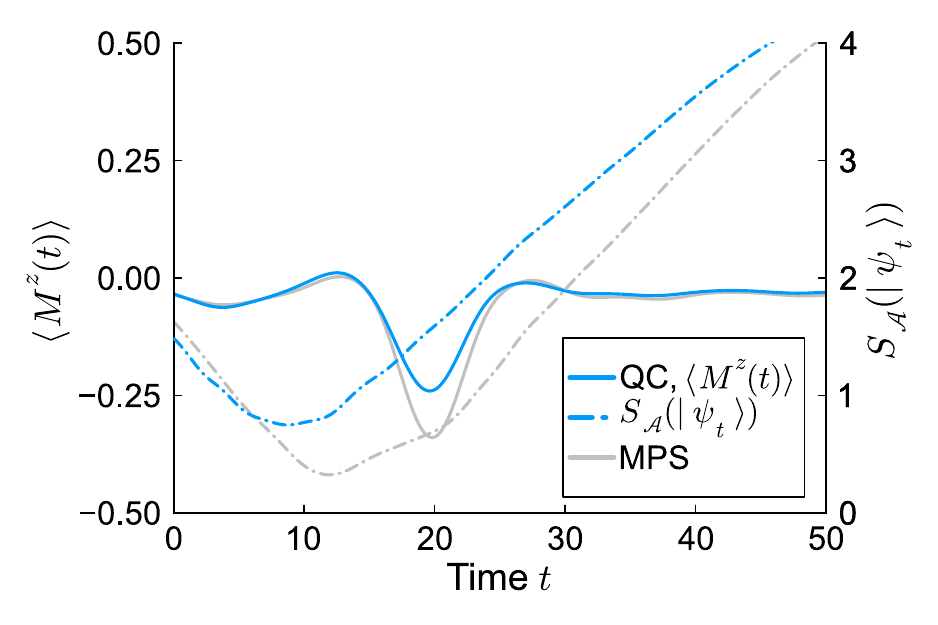}};
            
            \begin{scope}[x={(image.south east)},y={(image.north west)}]
            \node[anchor=north west] at (0.00, 0.97) {(b)};
            \end{scope}
        \end{tikzpicture}
    \end{minipage}
    \caption{(a) Schematic illustration of a quantum circuit that creates an approximate MPS $\ket{\psi_\mathrm{app}}$ with $L=5$ spins and $K=2$ layers. (b) Comparison of the time evolutions starting from the target MPS $\ket{\psi_0}$ (gray) and its quantum circuit (QC) approximation $\ket{\psi_\mathrm{app}}$ (blue) with $L=40$ for the mixed-field Ising chain~\eqref{eq:Ising}. The solid curve and the dash-dotted curve represent the time evolution of the expectation value $\langle M^z(t)\rangle$ and that of entanglement entropy $S_\mathcal{A}(\ket{\psi_t})$, respectively, where $\mathcal{A}=\{1,2,\ldots,\lfloor L/2 \rfloor\}$ denotes half of the system. For gray curves (same as in Fig.~\ref{fig:burst_plot}), the initial state $\ket{\psi_0}$ is obtained with $\tau=20$, $\chi=10$, $\beta=0.1$, and $\lambda_L=72/L^2$. For blue curves, the initial state $\ket{\psi_\mathrm{app}}$ is prepared by a staircase quantum circuit of $K=5$ layers that approximates the exact MPS $\ket{\psi_0}$. }
    \label{fig:burst_QC}
\end{figure}

\begin{figure}
    \centering
    \includegraphics[width=\linewidth]{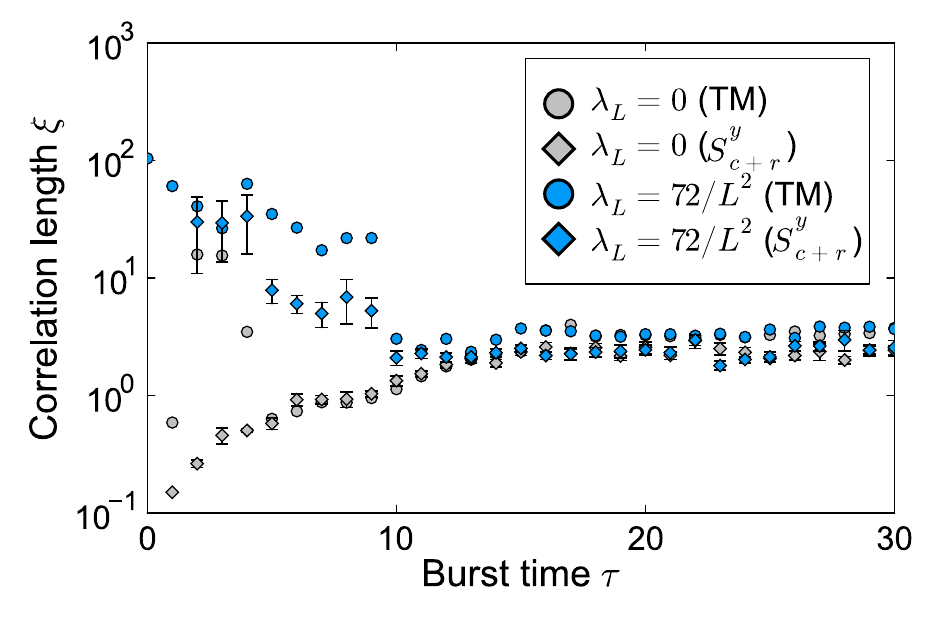}
    \caption{
        Correlation length $\xi$ of the initial state $\ket{\psi_0}$ versus the burst time $\tau$ for the mixed-field Ising chain~\eqref{eq:Ising} with $L=40$, $\chi=10$, and $\beta=0.1$. Colors denote $\lambda_L=72/L^2$ (blue) and $\lambda_L=0$ (gray). Circles correspond to the median of the correlation lengths (excluding five sites at each edge) derived from the local two-block transfer matrices (TMs). Diamonds represent estimates of $\xi$ with error bars obtained via a semi-log fit to the spatial decay of the connected $S_c^y$--$S_{c+r}^y$ correlation from the central site $c=\lfloor L/2 \rfloor$. Points for which the fit does not yield a positive correlation length are not shown.
    }
    \label{fig:corrlen}
\end{figure}

\begin{figure}
    \centering
    \hspace{-20pt}
    \includegraphics[width=\linewidth]{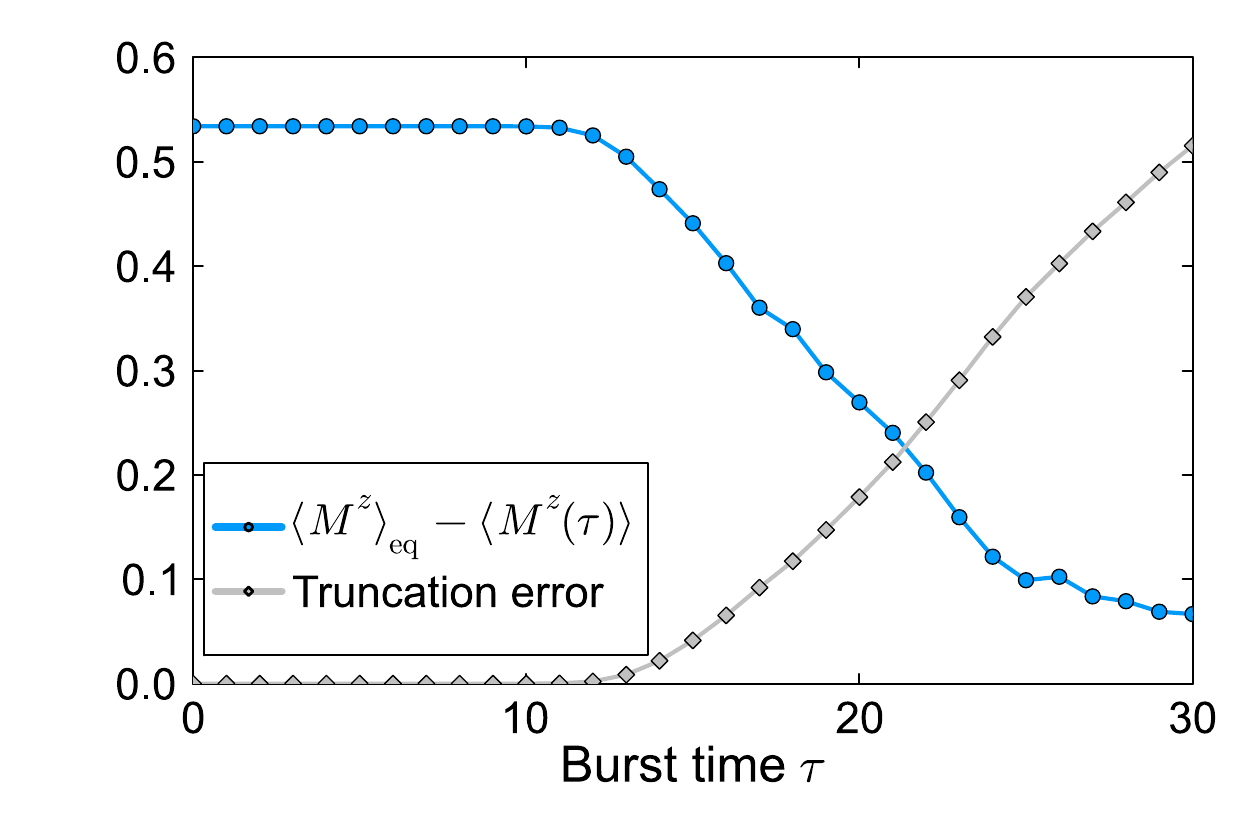}
    \caption{Burst amplitude $\langle M^z\rangle_\mathrm{eq}-\langle M^z(\tau) \rangle$ versus the burst time $\tau$ in the thermodynamic limit for the mixed-field Ising chain~\eqref{eq:Ising}. The initial state is obtained by truncating $e^{iH\tau}\ket{\downarrow}^{\otimes L}$ to bond dimension $\chi=10$, and the gray curve shows the average truncation error of the two bonds in an iMPS. }
    \label{fig:tauburst_Linf}
\end{figure}

\begin{figure}
    \centering
    \begin{minipage}{\linewidth}
        \centering
        \begin{tikzpicture}
            \node[anchor=south west, inner sep=0] (image) at (0,0) {\includegraphics[width=\linewidth]{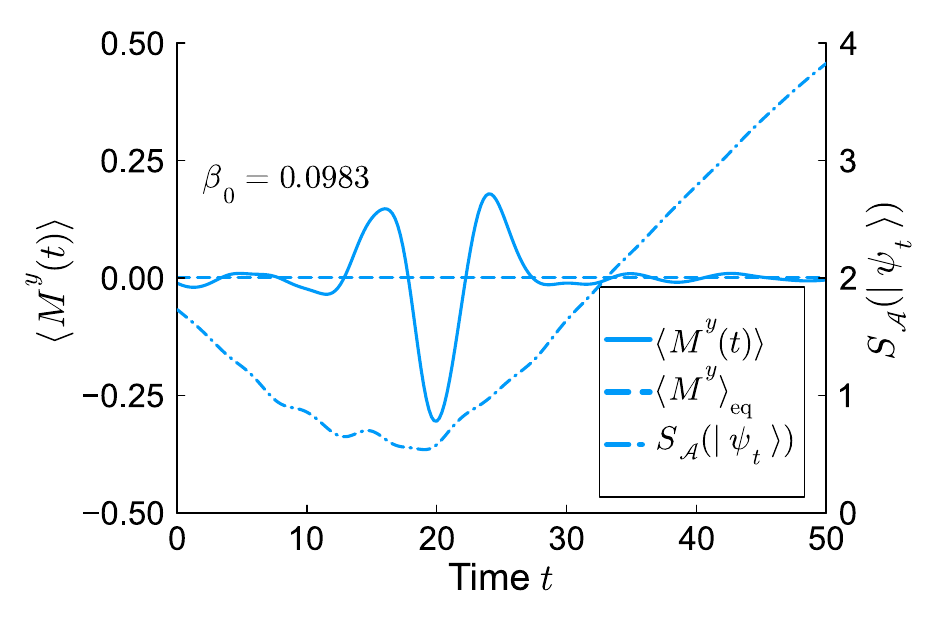}};
            \begin{scope}[x={(image.south east)},y={(image.north west)}]
            \node[anchor=north west] at (0.00, 0.97) {(a)};
            \end{scope}
        \end{tikzpicture}
    \end{minipage}
    \begin{minipage}{\linewidth}
        \centering
        \begin{tikzpicture}
            \node[anchor=south west, inner sep=0] (image) at (0,0) {\includegraphics[width=\linewidth]{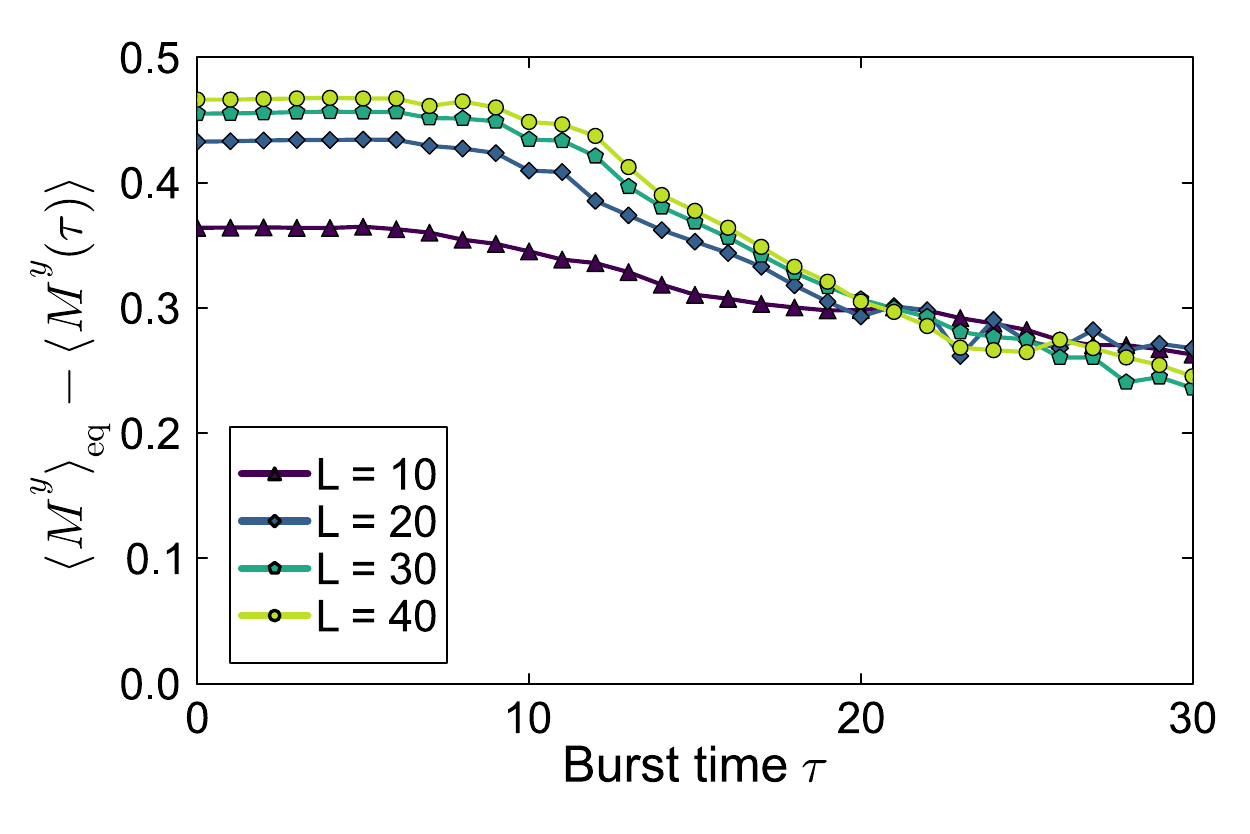}};
            \begin{scope}[x={(image.south east)},y={(image.north west)}]
            \node[anchor=north west] at (0.00, 0.97) {(b)};
            \end{scope}
        \end{tikzpicture}
    \end{minipage}
    \begin{minipage}{\linewidth}
        \centering
        \begin{tikzpicture}
            \node[anchor=south west, inner sep=0] (image) at (0,0) {\includegraphics[width=\linewidth]{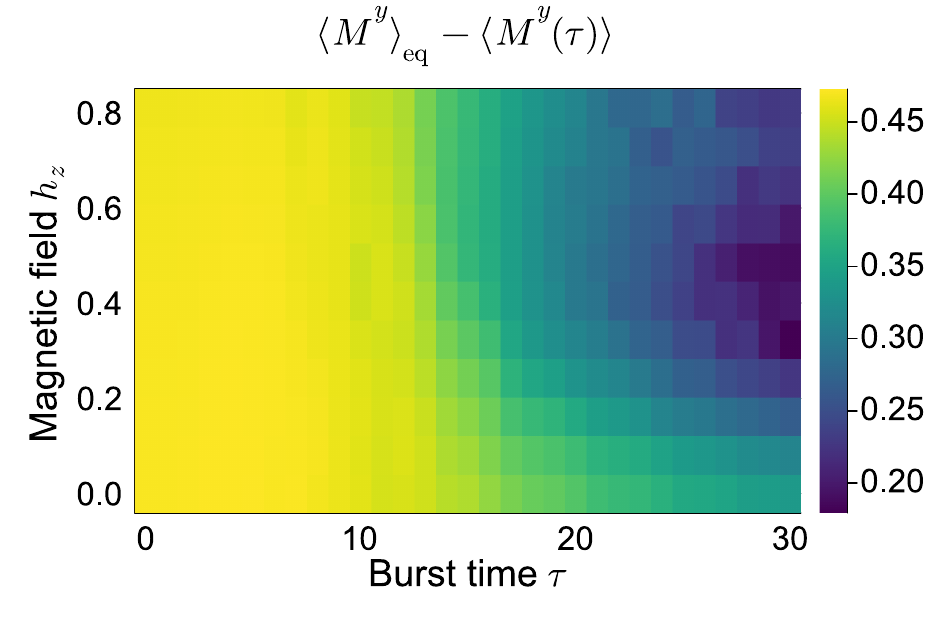}};
            \begin{scope}[x={(image.south east)},y={(image.north west)}]
            \node[anchor=north west] at (0.00, 0.97) {(c)};
            \end{scope}
        \end{tikzpicture}
    \end{minipage}
    \caption{
        (a), (b), and (c) show the same quantities as Fig.~\ref{fig:burst_plot}, Fig.~\ref{fig:Ltau}(a), and Fig.~\ref{fig:Ltau}(b), respectively, with $O=M^z$ replaced by $O=M^y$. Note that $\langle M^y\rangle_\mathrm{eq}=0$.
    }
    \label{fig:burst_Magy}
\end{figure}

\paragraph{Preparing initial states with quantum circuits}
An initial MPS that creates a burst at a given time can be approximately prepared by a shallow quantum circuit consisting only of neighboring two-qubit gates~\cite{Ran2020_encoding,Shirakawa2024_automatic,Rudolph2024_decomposition}.
Here we adopt the Iter$[D_i O_{\mathrm{all}}]$ method in Ref.~\cite{Rudolph2024_decomposition} to construct a staircase quantum circuit with $K$ layers (Fig.~\ref{fig:burst_QC}(a)) that gives an approximate state $\ket{\psi_\mathrm{app}}$ of a target MPS $\ket{\psi_0}$.

Figure~\ref{fig:burst_QC}(b) presents a comparison of bursts starting from two initial states: the exact MPS (gray curves) and the approximate state (blue curves).
A burst can still be observed starting from the latter, although the burst amplitude becomes a little smaller.
Entanglement entropy still shows a decrease starting from the approximate state in contrast to a typical linear growth; however, the decrease is less pronounced than in the case of the exact MPS.

\paragraph{Correlation length of the initial state}
In the main text, we have demonstrated that a burst can be created from a low-entangled state.
However, a low-entangled state is not necessarily easy to prepare when it has long-range correlations.
For instance, a GHZ state is an MPS of bond dimension $\chi=2$, but it is long-range correlated and susceptible to local perturbations.
To address this issue, we compute the correlation length $\xi$ of the initial state $\ket{\psi_0}$.

Figure \ref{fig:corrlen} shows the correlation length $\xi$ of the initial state $\ket{\psi_0}$ for $\lambda_L=72/L^2$ (blue) and $\lambda_L=0$ (gray), obtained via two different methods: a transfer-matrix method (circles) and a semi-log fit to the spatial decay of $S^y$ correlation (diamonds).
While the correlation length is large and unstable for $\lambda_L=72/L^2$ at a small $\tau$, it becomes short and stable in all of the four cases for $\tau\gtrsim 10$.
This result rules out the possibility that an initial state such as the one in Fig.~\ref{fig:burst_plot} has long-range correlations, which would otherwise make it highly susceptible to local perturbations.

\paragraph{Burst in the thermodynamic limit}
We demonstrate that a burst survives in the thermodynamic limit $L\to \infty$ by using the infinite matrix product state (iMPS) formalism~\cite{Vidal2007_classical,Orus2008_infinite}.
Here, we obtain an initial iMPS by computing $e^{iH\tau}\ket{\Psi_\mathrm{GS}}$ and truncating it to a bond dimension $\chi$, where $\ket{\Psi_\mathrm{GS}}$ is the ground state of $O$.
The first step of time-reversed evolution~\cite{Peres1984_stability} is also employed in the protocol of the echo dynamics~\cite{Fine2014_absence,Schmitt2018_irreversible,Dabelow2020_persistent}, which subsequently perturbs the state and lets it evolve in time according to the Hamiltonian $H$.
However, instead of perturbing it, we truncate the state to obtain the low-entangled initial state.
Although this method is simpler and, in principle, yields a smaller burst amplitude than the DMRG method with $\lambda_L=0$~\cite{method-comparison}, it is sufficient to demonstrate the {\it existence} of a burst.

The dependence of the burst amplitude on the burst time $\tau$ is shown in Fig.~\ref{fig:tauburst_Linf}.
For $\tau \lesssim 11$, the burst amplitude $\langle M^z\rangle_\mathrm{eq}-\langle M^z(\tau) \rangle$ is initially large and remains nearly constant.
This plateau is consistent with the observation that the state $e^{iH\tau}\ket{\downarrow}^{\otimes L}$ is well approximated by an iMPS of bond dimension $\chi=10$ in this time regime.
Then, the burst amplitude decays gradually for $\tau\gtrsim 11$.
This is accompanied by an increase in the truncation error, indicating that $e^{iH\tau}\ket{\downarrow}^{\otimes L}$ becomes increasingly complex, eventually exceeding the representational capacity of an iMPS.
These behaviors are similar to those shown in Fig.~\ref{fig:Ltau}.
Our results demonstrate that a burst initially survives in the thermodynamic limit, but it gradually decays as quantum scrambling becomes stronger.

\paragraph{Burst of $M^y$}
In the main text, we have focused on the burst of $M^z=L^{-1}\sum_{i=1}^L S_i^z$.
Here, we show that a burst can also be observed for $M^y=L^{-1}\sum_{i=1}^L S_i^y$.
In this case, diagonal elements with respect to the energy eigenbasis of the Hamiltonian~\eqref{eq:Ising} vanish due to time-reversal symmetry, and therefore its thermal equilibrium value is always zero.

As demonstrated in Fig.~\ref{fig:burst_Magy}(a), a burst of $M^y$ occurs at $\tau=20$ accompanied by a slow or negative increase in entanglement entropy.
Figure~\ref{fig:burst_Magy}(b) shows that the burst amplitude of $M^y$ remains large with increasing $L$, but it gradually decays with increasing $\tau$.
These behaviors are similar to those of $M^z$ shown in Figs.~\ref{fig:burst_plot} and \ref{fig:Ltau}(a).
As for the dependence on the longitudinal magnetic field $h_z$, Fig.~\ref{fig:burst_Magy}(c) shows that the burst decays faster as $h_z$ moves away from the integrable point $h_z=0$, which is similar to Fig.~\ref{fig:Ltau}(b).
On the other hand, the burst amplitude does not change significantly with $h_z$ for a small $\tau$, which may be attributed to the fact that the energy of the ground state of $M^y$ is zero and independent of $h_z$.
These results demonstrate that the burst phenomenon is not restricted to special Hamiltonian parameters or observables.

\clearpage
\setcounter{equation}{0}
\setcounter{figure}{0}
\setcounter{section}{0}
\setcounter{table}{0}

\renewcommand{\theequation}{S-\arabic{equation}}
\renewcommand{\thefigure}{S-\arabic{figure}}

\onecolumngrid

\begin{center}
    {\large \bf Supplemental Material}
\end{center}

\vspace{10pt}

In this Supplemental Material, we derive Eqs.~\eqref{eq:LRProb} and~\eqref{eq:LRMPS} in the main text for the local dimension $d\ge 2$, the system size $L\ge 2$, and the number of gates $s\ge 1$.
By definition, the bond dimension satisfies $\chi\ge1$.

\section{Concentration of measure for local random circuits: Proof of Eq.~\eqref{eq:LRProb}}
\label{Appendix:ConcLR}

\begin{dfn}[Definition 1 of \cite{Brandao2016_local}]\label{dfn:opnorm}
    Let $D$ be the dimension of the Hilbert space, and $\mathrm{H}$ be the Haar measure on $\mathbb{U}(D)$ (the group of $D\times D$ unitary matrices).
    Then, for a distribution $\nu$ on $\mathbb{U}(D)$ and a positive integer $k\in \mathbb{Z}_{>0}$, define
    \begin{align}
        g(\nu,k)\coloneqq \|\mathbb{E}_{U\sim \nu}U^{\otimes k,k}-\mathbb{E}_{U\sim \mathrm{H}}U^{\otimes k,k}\|_\infty,~U^{\otimes k,k}\coloneqq U^{\otimes k}\otimes (U^*)^{\otimes k}.\label{eq:Appendix2norm}
    \end{align}
\end{dfn}

\begin{lem}[Theorem 5 of \cite{Brandao2016_local}]\label{lem:LRdesign}
    For $\nu_{d,L}^{*s}$ defined in the main text, the following inequality holds:
    \begin{align}
        g(\nu_{d,L}^{*1},k)\le 1- (42500L\lceil \log_d(4k)\rceil^2 d^2 k^{5+3.1/\ln d})^{-1}.
    \end{align}
\end{lem}

\begin{lem}\label{lem:LRsufficient}
     For some positive numerical constant $C$,
    \begin{align}
        g(\nu_{d,L}^{*s}, k)\le d^{-4Lk},~k=\left\lfloor \qty(\frac{s}{CL^2 d^2\ln d})^{1/11}\right\rfloor
    \end{align}
    holds if $s\ge CL^2 d^2 \ln d$.
\end{lem}

\begin{proof}
    First, there exists a constant $C$ such that
    \begin{align}
        C\ge 170000\{\log_2(4k)+1\}^2 k^{-5+3.1/\ln 2}
    \end{align}
    holds for all $k\ge 1$. (Note that $C=2.1\times10^{6}$ is sufficient.)
    Then, the following inequality holds for $d\ge 2$:
    \begin{align}
        s&\ge CL^2 d^2\ln d\cdot k^{11}\notag\\
        &\ge 170000L^2d^2\ln d\lceil \log_d(4k)\rceil^2 k^{6+3.1/\ln d}.
    \end{align}
    It follows from Def.~\ref{dfn:opnorm} that $g(\nu_{d,L}^{*s}, k)=g(\nu_{d,L}^{*1}, k)^s$ (which is Eq.~(3) in Ref.~\cite{Brandao2016_local}).
    Using this equation and the inequality $(1-x)^s\le e^{-sx}$ for $0\le x\le 1$ and $s>0$, we obtain
    \begin{align}
        g(\nu_{d,L}^{*s}, k)&\le \exp\qty{-s(42500L\lceil \log_d(4k)\rceil^2 d^2 k^{5+3.1/\ln d})^{-1}}\notag\\
        &\le d^{-4Lk}.
    \end{align}
\end{proof}

\begin{dfn}[Definition 2.2 of \cite{Low2009_large}]
    A monomial in elements of a unitary matrix $U$ is of degree $(k_1,k_2)$ if it contains $k_1$ conjugated elements and $k_2$ unconjugated elements.
    We call it a balanced monomial of degree $k$ if it is of degree $(k,k)$.
    A polynomial is of degree $k$ if it is a sum of balanced monomials of degree at most $k$.
\end{dfn}

\begin{dfn}[Definition 2.6 of \cite{Low2009_large}]
    $\nu$ is an $\epsilon$-approximate unitary $k$-design if, for all balanced monomials $M$ of degree $\le k$,
    \begin{align}
        \abs{\mathbb{E}_\nu M(U)-\mathbb{E}_\mathrm{H} M(U)}\le \frac{\epsilon}{D^k}.
    \end{align}
\end{dfn}

\begin{lem}[Lemma 2.2.14 of \cite{Low2010_pseudo-randomness}]\label{lem:2norm-monomial}
    When $g(\nu,k)\le\epsilon$ holds, $\nu$ is a $D^{3k}\epsilon$-approximate $k$-design.
\end{lem}

\begin{lem}[Theorem 1.2 of \cite{Low2009_large}]\label{lem:kdesign}
    Let $f$ be a polynomial of degree $K$.
    Let us introduce $f(U)=\sum_i \alpha_i M_i(U)$, where $M_i(U)$ are monomials and $\alpha(f)=\sum_i |\alpha_i|$.
    Suppose that $f$ has the following probability concentration:
    \begin{align}
        \mathrm{Pr}_{U\sim \mathrm{H}}\qty[|f(U)-\mu|\ge \delta]\le Ae^{-b\delta^2},
    \end{align}
    where $\mathrm{H}$ denotes the Haar measure.
    Let $\nu$ be an $\epsilon$-approximate $k$-design.
    Then we have
    \begin{align}
        \mathrm{Pr}_{U\sim \nu}\qty[|f(U)-\mu|\ge \delta]\le \delta^{-2m} \qty{A\qty(\frac{m}{b})^m+\frac{\epsilon}{D^k}(\alpha(f)+|\mu|)^{2m}}
    \end{align}
    for a positive integer $m$ with $2mK\le k$.
\end{lem}

\begin{prop}[Eq.~\eqref{eq:LRProb} in the main text]
    \begin{align}
        \mathrm{Pr}_{U\sim \nu^{*s}_{d,L}}\qty[\abs{\mel{\psi}{U^\dagger O U}{\psi}-\frac{\Tr[O]}{d^L}}\ge \Delta_O a]\le 2\qty(\frac{m}{d^L a^2})^m\label{eq:AppendixLR}
    \end{align}
    holds for an integer $m\in\mathbb{Z}_{\ge0}$ that satisfies
    \begin{align}
        m\le \frac{k}{2},~k=\left\lfloor \qty(\frac{s}{CL^2 d^2\ln d})^{1/11}\right\rfloor,
    \end{align}
    where $C$ is the constant defined in Lem.~\ref{lem:LRsufficient}.
\end{prop}

\begin{proof}
    
    Without loss of generality, we assume that $O$ is traceless and $\Delta_O=1$.
    By Levy's lemma~\cite{Aubrun2026_optimal}, we have
    \begin{align}
        \mathrm{Pr}_{U\sim \mathrm{H}}\qty[\abs{\mel{\psi}{U^\dagger O U}{\psi}}\ge a]\le e^{-d^L a^2}.
    \end{align}
    We introduce $f_O(U)=\mel{\psi}{U^\dagger OU}{\psi}$ and apply Lem.~\ref{lem:kdesign} with $D=d^L$, $K=1$,
    and $\mu=0$.
    Also, $\alpha(f_O)= \sum_{i,j}|O_{i,j}|\le d^L \|O\|_F\le d^{3L/2}$ holds from the Cauchy-Schwarz inequality, where $\|X\|_F=\sqrt{\Tr[X^\dagger X]}$ is the Frobenius norm.
    It follows from Lems.~\ref{lem:LRsufficient} and~\ref{lem:2norm-monomial} that $\nu_{d,L}^{*s}$ is a $d^{-Lk}$-approximate $k$-design.
    Therefore, for $1\le m\le k/2$, we obtain
    \begin{align}
        \mathrm{Pr}_{U\sim \nu_{d,L}^{*s}}\qty[\abs{\mel{\psi}{U^\dagger O U}{\psi}}\ge a]&\le a^{-2m}\qty{\qty(\frac{m}{d^L})^m+(d^L)^{3m-2k}}\notag\\
        &\le 2\qty(\frac{m}{d^La^2})^m.
    \end{align}
    The inequality holds trivially for $m=0$ if we define $0^0=1$.
\end{proof}

\begin{rem}
    The right-hand side of Eq.~\eqref{eq:AppendixLR} takes its minimal value $\simeq 2 e^{-d^La^2/e}$ at $m=m_0\simeq d^L a^2/e$.
    Therefore, when $k\ge 2m_0$, the upper bound becomes double-exponentially suppressed with respect to $L$.
\end{rem}

\section{Covering number of \texorpdfstring{$\epsilon$}{ε}-net for matrix product states}
\label{Appendix:MPSnet}

We consider $L$ sites with each local dimension $d$ on a one-dimensional lattice.
Let $A=(A_i^{\sigma_i})_{i,\sigma_i}$ be a tensor, where $A_i^{\sigma_i}$ is a $D_i\times D_{i+1}$ complex matrix with $D_i=\min\{d^{i-1}, d^{L-i+1}, \chi\}$.
We define $\mathcal{T}_\chi$ as a set of such tensors and $\mathcal{S}_\chi$ as a set of $A\in \mathcal{T}_\chi$ in the following canonical form:
\begin{align}
    \mathcal{S}_\chi\coloneqq\qty{A\in \mathcal{T}_\chi \mid \forall i\in \{1,\ldots, L\}~\sum_{\sigma=1}^d A_i^{\sigma}A_i^{\sigma \dagger}=I_{D_i}}.
\end{align}
Let $\mathcal{M}_\chi$ be a set of normalized (OBC-)MPSs with a bond dimension less than or equal to $\chi$.
We define a map $\ket{\psi_\bullet}$ from a tensor to a state:
\begin{align}
    \ket{\psi_\bullet}:~\mathcal{S}_\chi\to\mathcal{M}_\chi,~A\mapsto \sum_{\{\sigma_i\}}A_1^{\sigma_1}\cdots A_L^{\sigma_L}\ket{\sigma_1\cdots\sigma_L}.
\end{align}
Note that an MPS in the canonical form is already normalized.
The map $\ket{\psi_\bullet}$ is surjective since any quantum state $\ket{\Psi}\in\mathcal{M}_\chi$ can be represented in the canonical form~\cite{Perez-Garcia2007_matrix}.

For a tensor $A\in \mathcal{T}_\chi$, we define its norm as follows:
\begin{align}
    \|A\|\coloneqq\sqrt{\sum_{i=1}^L \sum_{\sigma_i=1}^d \|A_i^{\sigma_i}\|_F^2}.
\end{align}
Then, any tensor $(A_i^{\sigma})_{\sigma=1}^d$ in the canonical form satisfies the following property:
\begin{align}
    \sum_{\sigma=1}^d \norm{XA_i^{\sigma}}_F^2=\norm{X}_F^2.\label{eq:can-inv}
\end{align}
Substituting $X=I_{D_i}$ in Eq.~\eqref{eq:can-inv}, we have
\begin{align}
    A\in \mathcal{S}_\chi\Rightarrow \|A\|\le R_\chi\coloneqq \sqrt{L\chi}.\label{eq:Rchi}
\end{align}

Now, we show the following lemma.
\begin{lem}\label{lem:Lipschitz}
    The map $\ket{\psi_\bullet}$ is $\sqrt{L}$-Lipschitz, that is,
    \begin{align}
        \forall A,B\in\mathcal{S}_\chi,~\|\ket{\psi_A}-\ket{\psi_B}\|\le \sqrt{L}\|A-B\|.
    \end{align}
\end{lem}
\begin{proof}
    By employing the telescoping sum, we have
    \begin{align}
        &~\norm{\ket{\psi_A}-\ket{\psi_B}}\notag\\
        =&~\norm{\sum_{\{\sigma_i\}} (A_1^{\sigma_1}\cdots A_L^{\sigma_L}- B_1^{\sigma_1}\cdots B_L^{\sigma_L})\ket{\sigma_1\cdots\sigma_L}}\notag\\
        =&~\norm{\sum_{i=1}^L\sum_{\{\sigma_i\}} A_1^{\sigma_1}\cdots A_{i-1}^{\sigma_{i-1}}(A_i^{\sigma_i}-B_i^{\sigma_i})B_{i+1}^{\sigma_{i+1}}\cdots B_L^{\sigma_L}\ket{\sigma_1\cdots\sigma_L}}\notag\\
        \le&~ \sum_{i=1}^L \norm{\sum_{\{\sigma_i\}} A_1^{\sigma_1}\cdots A_{i-1}^{\sigma_{i-1}}(A_i^{\sigma_i}-B_i^{\sigma_i})B_{i+1}^{\sigma_{i+1}}\cdots B_L^{\sigma_L}\ket{\sigma_1\cdots\sigma_L}}.
    \end{align}
    Let us define
    \begin{align}
        \Delta_i^{\sigma_i}&\coloneqq A_i^{\sigma_i}-B_i^{\sigma_i},\\
        \ket{\delta_i \psi}&\coloneqq \sum_{\{\sigma_i\}}A_1^{\sigma_1}\cdots A_{i-1}^{\sigma_{i-1}}\Delta_i^{\sigma_i}B_{i+1}^{\sigma_{i+1}}\cdots B_L^{\sigma_L}\ket{\sigma_1\cdots\sigma_L}.
    \end{align}
    Then it suffices to evaluate $\sum_{i=1}^L \norm{\ket{\delta_i \psi}}$.
    By repeatedly using Eq.~\eqref{eq:can-inv}, we obtain
    \begin{align}
        \norm{\ket{\delta_i\psi}}^2&=\sum_{\{\sigma_i\}}\norm{A_1^{\sigma_1}\cdots A_{i-1}^{\sigma_{i-1}}\Delta_i^{\sigma_i}B_{i+1}^{\sigma_{i+1}}\cdots B_L^{\sigma_L}}^2_F\notag\\
        &= \sum_{\sigma_1,\ldots,\sigma_i}\norm{A_1^{\sigma_1}\cdots A_{i-1}^{\sigma_{i-1}}\Delta_i^{\sigma_i}}^2_F.
    \end{align}
    Then, we define
    \begin{align}
        P_i\coloneqq\sum_{\sigma_1,\ldots,\sigma_{i-1}}(A_1^{\sigma_1}\cdots A_{i-1}^{\sigma_{i-1}})^\dagger (A_1^{\sigma_1}\cdots A_{i-1}^{\sigma_{i-1}}).\label{eq:P_i}
    \end{align}
    By construction, $P_i$ is positive semidefinite, and from Eq.~\eqref{eq:can-inv} we have $\Tr[P_i]=1$.
    Therefore, $O_{D_i}\preceq P_i \preceq I_{D_i}$ holds, where $X\preceq Y$ means that $Y-X$ is positive semidefinite.
    It follows that
    \begin{align}
        \norm{\ket{\delta_i\psi}}^2&=\sum_{\sigma_i} \Tr\qty[\Delta_i^{\sigma_i\dagger}P_i\Delta_i^{\sigma_i}]\notag\\
        &\le \sum_{\sigma_i} \Tr[\Delta_i^{\sigma_i\dagger}\Delta_i^{\sigma_i}]\notag\\
        &=\sum_{\sigma_i}\norm{\Delta_i^{\sigma_i}}_F^2.
    \end{align}
    Finally, by using the Cauchy-Schwarz inequality, we obtain
    \begin{align}
        \norm{\ket{\psi_A}-\ket{\psi_B}}&\le\sqrt{L}\sqrt{\sum_{i=1}^L \norm{\ket{\delta_i \psi}}^2}\notag\\
        &\le\sqrt{L}\sqrt{\sum_{i=1}^L\sum_{\sigma_i}\norm{\Delta_i^{\sigma_i}}_F^2}\notag\\
        &=\sqrt{L}\|A-B\|.\label{eq:Lipschitz}
    \end{align}
\end{proof}

We fix $\epsilon\in(0,2]$ and define $\mathcal{N}_{\chi,\epsilon}$ to be one of the maximal $\epsilon/\sqrt{L}$-separated subsets of $\mathcal{S}_\chi$, that is, for any $P,Q\in\mathcal{N}_{\chi,\epsilon}$ with $P\neq Q$, we have $\|P-Q\|>\epsilon/\sqrt{L}$.
Now, for any $\ket{\Psi}\in\mathcal{M}_\chi$ there exists $A\in \mathcal{S}_\chi$ such that $\ket{\psi_A}=\ket{\Psi}$ due to the surjectivity of $\ket{\psi_\bullet}$.
Then, there exists $A'\in\mathcal{N}_{\chi,\epsilon}$ such that $\|A-A'\|\le \epsilon/\sqrt{L}$ because otherwise $\mathcal{N}_{\chi,\epsilon}\cup\{A\}$ would become a larger $\epsilon/\sqrt{L}$-separated set of $\mathcal{S}_\chi$.
It follows from Lem.~\ref{lem:Lipschitz} that $\|\ket{\Psi}-\ket{\psi_{A'}}\|=\|\ket{\psi_A}-\ket{\psi_{A'}}\|\le \epsilon$.
Therefore, $\ket{\psi_{\mathcal{N}_{\chi,\epsilon}}}\coloneqq \{\ket{\psi_A}\mid A\in \mathcal{N}_{\chi,\epsilon}\}$ is an $\epsilon$-net of $\mathcal{M}_\chi$.

A tensor $A=(A_i^{\sigma_i})_{i,\sigma_i}\in\mathcal{T}_\chi$ has $d\sum_{i=1}^L D_i D_{i+1}$ complex matrix elements.
For each $A\in\mathcal{T}_\chi$, we define an $f_\chi\coloneqq 2d\sum_{i=1}^L D_i D_{i+1}$-dimensional real vector $\bm{v}_A$ whose elements consist of the real and imaginary parts of the matrix elements of $A$.
Let $B(\bm{v},r)$ be an $f_\chi$-dimensional ball of radius $r$ which is centered at $\bm{v}$.
Then, it follows from $\|\bm{v}_A\|=\|A\|$ and Eq.~\eqref{eq:Rchi} that
\begin{align}
    A\in \mathcal{S}_\chi\implies \bm{v}_A\in B(\bm{0},R_\chi).
\end{align}
Since balls in the set $\{B(\bm{v}_A,\epsilon/2\sqrt{L})\mid A\in \mathcal{N}_{\chi,\epsilon}\}$ are disjoint and contained in the ball $B(\bm{0},R_\chi+\epsilon/2\sqrt{L})$, we have by a volume argument
\begin{align}
    |\mathcal{N}_{\chi,\epsilon}|&\le \frac{(R_\chi+\epsilon/2\sqrt{L})^{f_\chi}}{(\epsilon/2\sqrt{L})^{f_\chi}}.
\end{align}
It follows from $\abs{\ket{\psi_{\mathcal{N}_{\chi,\epsilon}}}}\le |\mathcal{N}_{\chi,\epsilon}|$ and $f_\chi\le 2dL\chi^2$ that 
\begin{align}
    \abs{\ket{\psi_{\mathcal{N_{\chi,\epsilon}}}}}\le \qty(\frac{3L\sqrt{\chi}}{\epsilon})^{2dL\chi^2}.\label{eq:AppendixMPSnet}
\end{align}

\section{Derivation of Eq.~\eqref{eq:LRMPS}}

We denote
\begin{align}
    P_{\chi,a}\coloneqq \mathrm{Pr}_{U\sim \nu_{d,L}^{*s}}\qty[\sup_{\ket{\psi}\in\mathcal{M}_\chi}\abs{\mel{\psi}{U^\dagger O U}{\psi}-\frac{\Tr[O]}{d^L}}\ge \Delta_O a]
\end{align}
as the probability of finding a burst that is larger than or equal to $\Delta_Oa$ starting from one of the pure states in $\mathcal{M}_\chi$.
Using $|\mel{\psi}{O}{\psi}-\mel{\psi'}{O}{\psi'}|\le \Delta_O \|\ket{\psi}-\ket{\psi'}\|$ and applying the union bound to Eqs.~\eqref{eq:AppendixLR} and \eqref{eq:AppendixMPSnet}, we have
\begin{align}
    P_{\chi,a}&\le \mathrm{Pr}_{U\sim \nu^{*s}_{d,L}}\qty[\exists \ket{\psi'}\in\ket{\psi_{\mathcal{N}_{\chi,\epsilon}}}~\mathrm{s.t.}~ \abs{\mel{\psi'}{U^\dagger O U}{\psi'}-\frac{\Tr[O]}{d^L}}\ge \Delta_O (a-\epsilon)]\notag\\
    &\le\abs{\ket{\psi_{\mathcal{N}_{\chi,\epsilon}}}}\max_{\ket{\psi'}\in\ket{\psi_{\mathcal{N}_{\chi,\epsilon}}}}\mathrm{Pr}_{U\sim \nu^{*s}_{d,L}}\qty[\abs{\mel{\psi'}{U^\dagger O U}{\psi'}-\frac{\Tr[O]}{d^L}}\ge \Delta_O (a-\epsilon)]\notag\\
    &\le 2\qty(\frac{3L\sqrt{\chi}}{\epsilon})^{2dL\chi^2}\qty{\frac{m}{d^L(a-\epsilon)^2}}^m
\label{eq:Unionbound}
\end{align}
for all $0<\epsilon<a$.
By minimizing the right-hand side with respect to $\epsilon$, we obtain
\begin{align}
    P_{\chi,a}&\le 2\qty(\frac{3}{d\chi^{3/2}})^{2dL\chi^2}\qty(\frac{1}{md^L})^m\qty(\frac{dL\chi^2+m}{a})^{2(dL\chi^2+m)}.
\end{align}
Since $d$ and $\chi$ do not depend on the system size, the logarithm of $P_{\chi,a}$ is approximately bounded as
\begin{align}
    \ln P_{\chi,a}\lesssim -m\ln(md^L)+2(dL\chi^2+m)\ln\frac{dL\chi^2+m}{a}.
\end{align}
Equation~\eqref{eq:LRMPS} in the main text is thus proved.

\end{document}